\documentclass[a4paper,11pt]{amsart}

\usepackage{upgreek}
\usepackage{ esint }
\usepackage{color}
\usepackage{comment}
\usepackage{amssymb}
\usepackage{amsfonts}

\usepackage{amscd}

\usepackage{slashed}
\usepackage{pdflscape}
\usepackage{tikz}
\usepackage{enumerate}
\usepackage{multirow}
\usepackage{stmaryrd}

\usepackage{bbold}
\usepackage[linkcolor=blue]{hyperref}
\usepackage{mathrsfs}

\newcommand{\bm}[1]{\mbox{\boldmath $ #1 $}}

\newtheorem{theorem}{Theorem}[section]
\newtheorem{lemma}[theorem]{Lemma}
  
\newtheorem{proposition}[theorem]{Proposition}

\theoremstyle{definition}
\newtheorem{definition}[theorem]{Definition}
\newtheorem{example}[theorem]{Example}

\theoremstyle{remark}
\newtheorem{remark}[theorem]{Remark}

\usepackage{pdfsync}

\usepackage{graphicx} 
\usepackage{amsmath} 
\usepackage{amsfonts}
\usepackage{amssymb}

\newcommand{\be}{\begin{equation}}

\newcommand{\ee}{\end{equation}}

\newcommand{\Odane}{\O}

\newcommand{\II}{{ I\hspace{-.8mm}I}}
\newcommand{\IIo}{\mathring{\!{ I\hspace{-.8mm} I}}{\hspace{.2mm}}}

\newcommand{\ba}{\begin{array}}

\newcommand{\ea}{\end{array}}

\newcommand{\beq}{\begin{eqnarray}}

\newcommand{\eeq}{\end{eqnarray}}

\newtheorem{lm}{lemma}

\newtheorem{thee}{theorem}

\newtheorem{proo}{proposition}

\newtheorem{co}{corollary}

\newtheorem{rem}{remark}

\newtheorem{deff}{definition}

\newcommand{\bd}{\begin{deff}}

\newcommand{\ed}{\end{deff}}

\newcommand{\bl}{\begin{lm}}

\newcommand{\el}{\end{lm}}

\newcommand{\bp}{\begin{proo}}

\newcommand{\ep}{\end{proo}}

\newcommand{\bt}{\begin{thee}}

\newcommand{\et}{\end{thee}}

\newcommand{\bc}{\begin{co}}

\newcommand{\ec}{\end{co}}

\newcommand{\brm}{\begin{rem}}

\newcommand{\erm}{\end{rem}}

\usepackage{t1enc}

\newcommand{\newc}{\newcommand}

\renewcommand{\exp}{\operatorname{exp}}

\let\ccdot\cdot
\def\cdot{\hbox to 2.5pt{\hss$\ccdot$\hss}}

\newc{\aR}{\mbox{\boldmath{$ R$}}}
\newc{\aS}{\mbox{\boldmath{$ S$}}}
\newc{\aT}{\mbox{\boldmath{$ T$}}}
\newc{\aW}{\mbox{\boldmath{$ W$}}}

\newc{\aD}{\mbox{\boldmath{$ D$}}\hspace{-.2mm}}

\newc{\aK}{\mbox{\boldmath{$ K$}}}
\newc{\aL}{\mbox{\boldmath{$ L$}}}

\usepackage{amssymb}
\usepackage{amscd}

\newcommand{\Rho}{{\it P}}

\newcommand{\Ric}{{\it Ric}}
\newcommand{\Sc}{\it Sc}

\newcommand{\nn}[1]{(\ref{#1})}

\newcommand{\J}{{\mbox{\it J}}}

\newc{\obstrn}[2]{B^{#1}_{#2}}

\newcommand{\rpl}                         
{\mbox{$
\begin{picture}(12.7,8)(-.5,-1)
\put(0,0.2){$+$}
\put(4.2,2.8){\oval(8,8)[r]}
\end{picture}$}}

\newcommand{\lpl}                         
{\mbox{$
\begin{picture}(12.7,8)(-.5,-1)
\put(2,0.2){$+$}
\put(6.2,2.8){\oval(8,8)[l]}
\end{picture}$}}

\usepackage{ifthen}

\newc{\tensor}[1]{#1}
\newc{\Mvariable}[1]{\mbox{#1}}
\newc{\down}[1]{{}_{#1}}
\newc{\up}[1]{{}^{#1}}

\newc{\JulyStrut}{\rule{0mm}{6mm}}
\newc{\midtenPan}{\mbox{\sf S}}
\newc{\midten}{\mbox{\sf T}}
\newc{\midtenEi}{\mbox{\sf U}}
\newc{\ATen}{\mbox{\sf E}}
\newc{\BTen}{\mbox{\sf F}}
\newc{\CTen}{\mbox{\sf G}}

\def\sideremark#1{\ifvmode\leavevmode\fi\vadjust{\vbox to0pt{\vss
 \hbox to 0pt{\hskip\hsize\hskip1em
 \vbox{\hsize2cm\tiny\raggedright\pretolerance10000
  \noindent #1\hfill}\hss}\vbox to8pt{\vfil}\vss}}}

\newcommand{\edz}[1]{\sideremark{#1}}

\numberwithin{equation}{section}

\newcommand{\hh}{{\hspace{.3mm}}}

\newcommand{\nablab}{\bar\nabla}

\DeclareMathOperator{\divergence}{div}
\newcommand{\smallJ}{\scalebox{.7}{$\J$}}

\newcommand{\db}{{d-1}}

\newcommand{\sss}{\scriptscriptstyle}

\newcommand{\smidge}{{\hspace{-1mm}}}
\newcommand{\Smidge}{{\hspace{-.8mm}}}
\newcommand{\SSmidge}{{\hspace{-.2mm}}}

\renewcommand\geq{\geqslant}
\renewcommand\leq{\leqslant}

\DeclareMathOperator{\w}{w}
\DeclareMathOperator{\ext}{d}
\DeclareMathOperator{\D}{L}
\DeclareMathOperator{\s}{s}
\DeclareMathOperator{\I2}{S}
\DeclareMathOperator{\h}{h}
\DeclareMathOperator{\x}{x}
\DeclareMathOperator{\y}{y}
\DeclareMathOperator{\Vol}{Vol}

\DeclareMathOperator{\GJMS}{P}
\DeclareMathOperator{\Operator}{O}

\begin{document}

\renewcommand{\today}{}
\title{
{Renormalized Volume}\\[4mm]
}
\author{ A. Rod Gover${}^\sharp$ \& Andrew Waldron${}^\natural$}

\address{${}^\sharp$Department of Mathematics\\
  The University of Auckland\\
  Private Bag 92019\\
  Auckland 1142\\
  New Zealand
  } \email{r.gover@auckland.ac.nz}
  
  \address{${}^{\natural}$Department of Mathematics\\
  University of California\\
  Davis, CA95616, USA} \email{wally@math.ucdavis.edu}

\vspace{10pt}

\renewcommand{\arraystretch}{1}

\begin{abstract} 

We develop a universal distributional calculus for regulated volumes of metrics that are singular along hypersurfaces. When the hypersurface is a conformal infinity we give  simple integrated distribution expressions for the divergences and anomaly of the regulated volume functional valid for any choice of regulator.
For closed hypersurfaces or
conformally compact geometries,  methods from a previously developed boundary calculus for conformally compact manifolds can be applied to give explicit holographic formula\ae\   for the divergences and anomaly expressed as hypersurface integrals over local quantities (the method also extends to non-closed hypersurfaces).
The resulting anomaly does not depend on any particular choice of regulator, while the regulator dependence of the  
divergences is precisely captured by these formul\ae.
Conformal hypersurface invariants can be studied by demanding that the singular metric obey, smoothly and formally to a suitable order,
a Yamabe type problem with boundary data along the conformal infinity. We prove that the volume anomaly for these singular Yamabe solutions 
is a conformally invariant integral of a local~$Q$-curvature that generalizes the Branson~$Q$-curvature by including data of the embedding.
In each dimension this canonically defines a higher  dimensional generalization of the Willmore  energy/rigid string action.
Recently Graham proved that the first variation of the volume anomaly recovers the density obstructing smooth solutions to this singular Yamabe problem; we give a new proof of this result employing our boundary calculus. Physical applications of our results include studies of quantum corrections to entanglement entropies.

\vspace{10cm}


\end{abstract}


\maketitle

\pagestyle{myheadings} \markboth{Gover \& Waldron}{Renormalized Volume}

\newpage

\tableofcontents

\section{Introduction}

The problem of defining and computing volumes for 
manifolds with singular metrics
\begin{equation}\label{ds2}
ds^2= \frac{dx^2 + h(x)}{x^2}\, ,
\end{equation}
has played a central role in the anti de Sitter/conformal field theory (AdS/CFT) correspondence as well as in conformal geometry~\cite{Mal,AdSCFTreview,FGQ,GZ}. Volumes of regions approaching the hyper\-surface/bo\-u\-ndary~$\Sigma$ 
diverge at $x=0$. However, by a suitable cut-off and renormalization, a renormalized volume
functional can be defined that is invariant under conformal transformations of the bound\-ary metric $h$ up to a (conformally invariant) anomaly. An early and spectacular AdS/CFT success was the work of Henningson and Skenderis that identified this as the Weyl or trace anomaly of the boundary quantum field theory~\cite{Henningson}. Significant mathematical progress was made 
when Fefferman and Graham ~\cite{FGQ} showed that for Poincar\'e--Einstein structures (Euclidean signature, asymptotically AdS, Einstein manifolds), the renormalized volume anomaly recovered Branson's~$Q$-curvature~\cite{BQ} for the boundary manifold. This is an important  invariant of conformal geometries~(see~\cite{GJ,DM} and the reviews~\cite{BG,WhatQ}).
The renormalized volume is usually obtained by computing a Fefferman--Graham coordinate expansion of a bulk metric tensor solving, to some order, a bulk problem with boundary data at a conformal infinity~$\Sigma$. This expansion is  inserted first in the metric determinant and, in turn, into a regulated volume integral.
We shall present a 
general,  simplifying  and efficient approach to volume computations for singular metrics 
that, in contrast to previous studies, does not  rely on solving any particular bulk problem.

\medskip

Let $(M,g^o)$ be a Riemannian manifold whose metric $g^o$ is singular along an  hypersurface $\Sigma$. 
For simplicity we take all structures to be oriented. Given a compact region~$D$ such that $\partial D\cap\Sigma\neq \emptyset$, we define the regulated volume of $D$ as follows (see also the diagram in Display~\nn{iamnotclosed}).

\begin{definition}\label{def1}
Given $(M,g^o,D)$ as above, let $ \varepsilon \geq 0$ and $\Sigma_\varepsilon$ be a smooth, one parameter family of oriented hypersurfaces such
that 
\begin{enumerate}[(i)]
\item
$\Sigma_0=\Sigma$, 
\item $\Sigma_{\varepsilon>0} \cap \Sigma=\emptyset$, and
\item $\Sigma_{\varepsilon>0}$ separates $D$ into a disjoint union  $D=D_\varepsilon\cup (D\backslash D_\varepsilon)$, where $g^o$ is non-singular in $D_\varepsilon$. \end{enumerate}
Then the {\it regulated volume} of $D$ is defined to be
$$
\Vol_\varepsilon(D,\Sigma):=\int _{D_\varepsilon} \sqrt{\det g^o}\, .
$$
\end{definition}

Our methods can in principle  be applied to 
quite general metric singularities, but we focus  
on the  
mathematically and physically 
central 
conformally compact case for which the hypersurface $\Sigma$ is a conformal infinity for the metric $g^o$.
In this case the regulated volume may be expanded as a sum of divergences (poles in $\varepsilon$), an anomaly (a $\log \varepsilon$ term) and 
the~$\varepsilon$-independent renormalized volume plus $\mathcal O(\varepsilon)$ contributions. 
We give simple results for the divergences and  anomaly in terms of integrals over Dirac-delta distributions, and their derivatives,  depending on a defining function for the hypersurface. These results encode the precise dependence of the divergences on the choice of regulator~$\Sigma_\varepsilon$, while the anomaly is independent of the regulator and is conformally invariant (in a suitable sense).

For applications, results for the anomaly and divergences given as hypersurface integrals over local quantities are required. 
Here it is propitious to assume  
that the hypersurface $\partial D\cap \Sigma$ is closed. We also indicate how to handle non-closed boundaries in the current work, but reserve a detailed treatment to a sequel article.
The key tool for both cases is the boundary calculus for conformally compact manifolds developed in~\cite{GW,GLW}. For conformally compact structures, we present exact and explicit formulas for both the divergences and the anomaly in the regulated volume. These are expressed as boundary integrals over local quantities and hold for {\it any} regulator and {\it any} conformally compact manifold.
 
Our results can be applied to study  the
  conformal geometry of hypersurface embeddings.
  Quantities that depend only on the conformal embedding of the hypersurface~$\Sigma$, can be found and studied by
requiring that the metric $g^o$ 
solves a singular version of the Yamabe problem of finding conformally rescaled metrics with constant scalar curvature~\cite{CRMouncementCRM,GW15}. In fact, 
since a unique asymptotic  solution to the singular Yamabe problem exists (at least up to the order required for the anomaly) 
for any conformally compact stucture, there is a corresponding canonical result for the 
anomaly which   is given by an integral over a density that can be defined for any hypersurface in a Riemannian manifold; this gives a new $Q$-curvature that includes 
extrinsic curvature data. In particular, 
by construction, it only depends on the conformal data of how the hypersurface~$\Sigma$ is embedded in the bulk.

\medskip

Since we need not   impose the bulk Einstein equation, our results apply to general bulk/boundary problems and thus extend an important aspect of the AdS/CFT program.
  A second motivation for our study is that this general setting allows us to study the extrinsic conformal geometry of the boundary geometry. Mathematically, our results are part of a general program to understand conformal hypersurface geometry~\cite{GW15} (see~\cite{CG15} for an overview),
     and  to develop  the calculus for integrated conformal hypersurface invariants begun in~\cite{GGHW15}.
Indeed, we wish to   
  initiate a new approach to geometric invariant theory based on holographic renormalization.
This program is also of substantial physical interest:
Soon after the original AdS/CFT duality was proposed, Graham and Witten showed how 
the renormalized volume method could be extended to bulk minimal surfaces in order to analyze holographic observables for boundary submanifolds~\cite{GrahamWitten}. This study produced conformal hypersurface invariants, the most notable of which, perhaps,  is the  Willmore energy for surfaces embedded in 3-manifolds. More recently, classes of these observables have been related to entanglement entropies of boundary field theories~\cite{RT,Astaneh,Perlmutter}.

\medskip

A key observation underlying our approach is that the metric in Equation~\nn{ds2} is determined by the pair
$$
g=dx^2 + h(x) \ \mbox{ and } \ \sigma = x\, ,
$$
where $(g,\sigma)$ are a non-singular bulk metric and function. However, we equally well could have chosen the pair $(\Omega^2 g, \Omega \sigma)$ where $\Omega$ is any smooth, positive function of the bulk manifold. The equivalence 
$$
g\sim \Omega^2 g
$$
defines a conformal class of metrics $\bm c:=[\hh g\hh ]=[\hh\Omega^2g\hh]$ and suggests that conformal, rather than Riemannian, geometry is the correct tool for simultaneously handling  bulk 
and boundary geometries in an AdS/CFT 
setting. The  equivalence $(g,\sigma)\sim (\Omega^2 g,\Omega \sigma)$ defines a bulk, weight one, conformal density $\bm \sigma :=[\hh g\, ;\, \sigma]=[\hh \Omega^2 g\, ;\, \Omega \sigma]$.
When the function~$\sigma$ has a suitable non-empty, nowhere dense zero locus, the data~$(M,\bm \sigma)$ is called an {\it almost Riemannian geometry}~\cite{Goal}
 (note that  the canonical equivalence class representative $[\sigma^{-2}g\, ;\, 1]$ defines a singular Riemannian metric~$ds^2$ as in Equation~\nn{ds2}). 
When this zero locus~$\Sigma$ is a hypersurface or boundary component and 
the function~$\sigma$ is for it a defining function, then~$\Sigma$ is a {\it conformal infinity} for the singular metric $g/\sigma^2$. 
When $M$ is compact with boundary the zero locus of $\sigma$, then $(M,\bm \sigma)$ is said to be conformally compact. In fact, for our purposes, it suffices to work in a collar neighborhood of the boundary, therefore we shall say that $(M,\bm \sigma)$ is {\it conformally compact} in any case where 
 $\Sigma$ is closed.
Reformulating the renormalized volume problem in terms of almost Riemannian geometry brings to bear a potent boundary calculus of conformally compact manifolds that utilizes the bulk conformal structure~\cite{Goal,GW,GLW}.

One of our main results is that for 
{\it any} conformally compact manifold, the anomaly  is given  as an integral over the corresponding  extrinsically coupled $Q$-curvature first introduced in~\cite{GW}. 
When regulating a quantum field theory, a dimensionful scale must be introduced. 
A powerful way to handle dimensionful quantities is to use conformal densities. 
Physically, a dimensionful quantity, such as a length, will vary across spacetime if different choices of local unit systems
are employed.  For example, the invariant property of a length is its linear homogeneity under Weyl transformations. Hence to regulate renormalized volumes we introduce a nowhere vanishing, unit weight, bulk conformal density $\bm \tau$ and cut off the bulk geometry at a regulating surface $\Sigma_\varepsilon$ determined by 
$$
\bm \sigma/\bm \tau =\varepsilon\in {\mathbb R}_+\, .
$$
The renormalized volume anomaly is then given, in $d$ bulk dimensions, by a  boundary/hypersurface integral
$$
{\mathcal A}=\frac{1}{(d-1)!(d-2)!}\, \int_\Sigma \bm Q^{\scalebox{.7}{$\bm\sigma$}}\, ,
$$
where  $\bm Q^{\scalebox{.7}{$\bm\sigma$}}=[\hh g\, ;\, Q\, ]$ is an {\it extrinsically coupled $Q$-curvature} of~$\Sigma$ which generalizes the standard Branson~$Q$-curvature. When the singular metric is determined by the conformal hypersurface embedding through the singular Yamabe problem,  it has a  simple explicit formula$$
\bm Q := (-\D)^{d-1} \log \bm \tau \, \Big|_\Sigma\, .
$$
Equally compact formul\ae\  are available for the integrated, local coefficients of the~$\frac{1}{\varepsilon^k}$ ($d-1\geq k\geq 1$) divergences in the regulated volume; these necessarily depend on the choice of regulator~$\bm \tau$ and are proportional to 
$$
\int_\Sigma\, \D^{d-k-1}\Big(\frac1{\bm \tau^k}\Big)\, .
$$ 
Details are given in Sections~\ref{RV} and~\ref{singYam}, but the main features of these results are as follows:
\begin{itemize}
\item \label{shift} The quantity $\bm Q$ is a weight $1-d$ density and   is invariant under simultaneous conformal rescalings  $g\to \Omega^2g$ and $\tau\to \Omega \tau$.
Fixing a choice of regulator~$\tau$ and transforming only the metric, the $Q$-curvature then has the famous linear shift property
\begin{equation*}
Q\mapsto \Omega^{1-d} (Q- \GJMS_{d-1} \log\Omega)\, .
\end{equation*}
Here, $\GJMS_{d-1}$ is a so-called extrinsic conformal Laplacian power~\cite{GW15}, which is a canonical extrinsically coupled analog of the
conformally invariant 
GJMS operators of~\cite{GJMS}.
The quantity $\GJMS_{d-1} \log\Omega$ is a total divergence along~$\Sigma$, and hence the $Q$-curvature integrates to an invariant of the (closed) boundary conformal manifold.

\vspace{1mm}
\item The anomaly is in general  non-vanishing.
However, when the bulk geometry is Einstein, the extrinsic $Q$-curvature vanishes for odd dimensional $\Sigma$, while for even dimensional~$\Sigma$  it reduces to the standard $Q$-curvature of the boundary conformal geometry.

\vspace{2mm}
\item
The operator $\D$ is the so-called Laplace--Robin operator (see Section~\ref{LAPROB}) determined by the conformal unit defining density~$\bm \sigma$ (see Section~\ref{singYam}). Along the boundary~$\Sigma$ it is a conformally invariant Robin-type (Dirichlet plus Neumann) operator that controls conformally invariant boundary data for conformal infinities, while in the bulk it is a Laplace-type operator that generates wave equations for matter fields~\cite{GoSigma,GoverS,Shaukat,GLW}.

\vspace{1mm}
\item The Laplace--Robin operator forms part of an $\mathfrak{sl}(2)$ solution generating algebra~\cite{GW}; this is the key technical tool for our computations.

\vspace{1mm}
\item In dimension $d=3$, the anomaly is 
a sum of the  Euler characteristic  for 2-manifolds and the rigid string action/Willmore energy for embedded surfaces (see Equation~\nn{Willd=2}).

\vspace{1mm}
\item The simplicity of the integrands appearing in the above
 formul\ae\ for the anomaly and divergences  is achieved by expressing these as local bulk quantities restricted to the hypersurface. This type of bulk boundary correspondence often carries the moniker ``holography'', so
 expressions for  hypersurface invariants given by the restriction of  bulk quantities are termed 
  {\it holographic formul\ae}~\cite{GW}.

\vspace{1mm}
\item The above simple formul\ae\ for the  
extrinsically coupled   $Q$-curvature and divergences 
 rely on the existence of asymptotic solutions to a singular version of the Yamabe problem.  
 As already mentioned, there exist 
 also extremely simple distributional formul\ae\  for these quantities valid
both  for general singular metrics and
for non-closed~$\Sigma$; see Theorem~\ref{distdivsan}. For conformally compact structures
the local boundary integral expressions for these are given in Proposition~\ref{divs} and Theorem~\ref{anomalytheorem}.
\end{itemize}

Variational problems for $Q$-curvatures are also a subject of intense study. In particular, the metric variation of the Branson $Q$ curvature yields the Fefferman--Graham obstruction tensor~\cite{GraHi}. This latter quantity determines whether log terms must be introduced when solving Einstein's equations in a Fefferman--Graham expansion off a conformal infinity. For the extrinsically coupled $Q$-curvature, an analogous problem  is to treat variations of  the anomaly~${\mathcal A}$ with respect to variations of the embedding of the hypersurface~$\Sigma$. In~\cite{GGHW15}, an efficient calculus for this type of variation was developed by writing boundary energy functionals holographically in terms of bulk integrals. This is also a key part of  our extrinsic $Q$-curvature computation. Indeed the hypersurface variation of the anomaly  plays the role of an obstruction to smoothly solving a bulk problem, but rather than Einstein's equations, the relevant problem is the singular Yamabe problem. This problem was found to be obstructed in~\cite{ACF}, with the obstruction shown to be a 
non-trivial  conformal invariant of embedded surfaces when~$d=3$. Generally, the obstruction was shown to give a natural conformal hypersurface invariant and called the obstruction density in~\cite{GW15}.  Low dimensional examples are known to be variational~\cite{GGHW15}.
Very recently, Graham has proved that the obstruction density of~\cite{CRMouncementCRM,GW15} is the variation of the renormalized volume anomaly~\cite{Grahamnew}. In Section~4 we rederive this result within our framework. 

Our results can  be applied to the situation encountered in entanglement entropy studies where the relevant renormalized volume computation applies to the renormalized ``area'' of a minimal hypersurface in a (spatial) bulk geometry whose boundary is some (codimension two with respect to the spatial bulk geometry) closed hypersurface separating entangled spatial regions in a boundary quantum field theory. For that, one only needs to compute the induced metric along the minimal hypersurface and then treat the entangling hypersurface as the boundary for the minimal hypersurface.
The Laplace--Robin operator characterization of volume divergences is extremely simple, but naturally will produce complicated formul\ae\  in terms of both intrinsic and extrinsic curvatures when higher divergences in higher dimensions are considered.
However, 
since quantum corrections to holographic entanglement entropies are of current topical interest (see for example~\cite{Lewkowycz,Engelhardt}),
we have converted our compact Laplace--Robin-type formul\ae\  
into integrated local curvature expressions for the first four divergences; see  Equations~\nn{firsttwo} and~\nn{nnlo} and Appendix~\ref{nnnlo}.

 
 Many of our results
 were originally obtained using a tractor calculus approach~\cite{BEG},  and then rederived using conformal densities with a view to making the materially generally accessible. We refer the interested reader to our work~\cite{GW15} for further details in this direction.


\subsection{Geometry conventions}\label{conventions}

All structures will be assumed to be smooth ({\it i.e.}~$C^\infty$).
We work  with oriented manifolds~$M$ of dimension~$d$ and hypersurfaces in $M$, meaning
compatibly oriented, codimension~1 submanifolds embedded in~$M$.  When the dimension~$d$ equals 
 three or four,
 we often refer to the latter   as surfaces and spaces, respectively, and we will refer interchangeably to the manifold $M$ as the ``bulk/ambient/host'' manifold. 
 (Note that the exterior derivative will be denoted by $\ext$, to avoid confusion with the dimension~$d$.)
 When~$M$ is equipped with a Riemannian metric~$g$ (for simplicity we assume Euclidean signature), its Levi-Civita connection will be denoted by  
$\nabla$
or $\nabla_a$. 
The corresponding Riemann curvature tensor $R^g$  is
$$R(u,v)w=[\nabla_u,\nabla_v]w-\nabla_{[u,v]}w\, ,$$
for arbitrary vector fields $u$, $v$ and $w$ (we drop the superscript indicating the dependence on the metric $g$ on geometric quantities when this is clear by context). In an index notation, $R$  is denoted by
$R_{ab}{}^c{}_d$ and $R(u,v)w$ is $u^a v^b R_{ab}{}^c{}_d w^d$.  
Cotangent and tangent spaces will be canonically identified using the metric tensor~$g_{ab}$, meaning that  this will be used to raise and lower indices in the standard fashion.

The Riemann curvature can be decomposed into the 
trace-free {\it Weyl curvature} $W_{abcd}$ and the symmetric {\it Schouten tensor}~$\Rho_{ab}$ according to 
$$
R_{abcd}=W_{abcd}+ 2g_{a[c} \Rho_{d]b}-2g_{b[c} \Rho_{d]a}\, .
$$
Here antisymmetrization over a pair of indices is denoted by square brackets so that $X_{[ab]}:=\frac{1}{2}\big(X_{ab}-X_{ba}\big)$. The Schouten and Ricci tensors are related by
$$
\Ric_{bd}:=R_{ab}{}^a{}_d=(d-2)\Rho_{bd}+g_{ab}\J\, ,\quad \J:=\Rho_a^a\, .
$$
The scalar curvature $\Sc=g^{ab}\Ric_{ab}$, thus $\J=\Sc/(2(d-1))$. In two dimensions the Schouten tensor defined above is pure trace with $\J=\frac12 \Sc$. 

Given an embedded hypersurface~$\Sigma$, intrinsic analogs of the above geometric quantities will be decorated with  bars, so for example, the induced metric is $\bar g_{ab}$ and its Riemann tensor is $\bar R_{ab}{}^c{}_d$. The same indices are used  for hypersurface tensors as for those in the host space~$M$ (remembering, of course, that the former are orthogonal to the unit normal vector). Equalities that hold only along the hypersurface $\Sigma$ are denoted by $\stackrel\Sigma=$.

We use 
$|u|:=\sqrt{u_au^a{}^{\phantom 2}\!\!}:=\sqrt{u^2}$ to denote the length of a vector $u$. Symmetrization over groups of indices is indicated by round brackets, and the notation $(\cdots)\circ$ denotes the trace-free, symmetric  part of a group of indices.

\section{Mathematical background}

\subsection{Conformal densities}\label{confdenintro}

A {\it conformal manifold} is a $d$-manifold $M$ equipped with
a conformal class of metrics
$$
{\bm c}
:=[\hh g]=[\Omega^2 g]\, ,
$$
where $\Omega:=\exp(\varpi)$ is any smooth, strictly-positive function. On a conformal manifold, a {\it conformal density}  of weight $w\in{\mathbb R}$ is 
a equivalence class of (metric, function)
pairs
  defined by $${\bm \tau}:=[\hh g\, ;\, \tau]=[\Omega^2g\, ;\, \Omega^w \tau]\, .$$
In the following we use {\it density} as a moniker for  conformal density. 
A weight $w=0$ density is a function on $M$, in which case we may denote $[\hh g\, ;\, f\, ]$ by~$f$. Equal weight densities ${\bm f}=[\hh g\, ;\, f\, ]$ and ${\bm h}=[\hh g\, ;\, h\, ]$ may be added according to 
${\bm f}+{\bm h}:=[\hh g\, ;\, f+h\, ]$
yielding a density of the same weight, while multiplication ${\bm f}{\bm h}:=[\hh g\, ;\, fh\, ]$  yields a density with weight given by the sum of weights (here
${\bm f}$, ${\bm h}$ need not be equally weighted). The unit density is the weight $0$ density $1:=[\hh g\, ;\, 1\, ]$. Tensor-valued conformal densities can be defined analogously to their scalar counterparts.
For example, if ${\bm f}=[\hh g\, ; f]$ is a weight zero density then its {\it conformal gradient}
\begin{equation}\label{conformalgradient}
\bm \nabla_a {\bm f}:=[\hh g\, ; \, \nabla_a f]\, ,
\end{equation}
defines a weight zero covector-valued density.

When  $w=1$ and the function $\tau$ is strictly positive, we call ${\bm \tau}=[\hh g\, ;\, \tau\, ]$ a {\it true scale}, or  simply a ``scale'' 
(which dovetails nicely with its physical interpretation). A true scale canonically determines a Riemannian geometry~$(M,g^o_{ab})$ via the  equivalence class representative~$
{\bm \tau}=[\hh g^o\, ;\, 1\, ]
$. Conversely, given a true scale ${\bm \tau}$ and a density ${\bm f}$, this canonically determines a function $f$ by expressing ${\bm f}=[\hh g^o,f\, ]$. We will often perform computations involving densities in terms of such a function $f$ and term this ``working in a scale'', which we will label either by specifying a given metric~$g\in\bm c$ or a true scale $\bm \tau$. In contexts where  the choice of scale/metric 
is clear, we will use unbolded symbols for the corresponding equivalence class representatives for densities.

Given a unit weight density ${\bm \sigma}$ (which need not be a true scale) and a weight $w$ density ${\bm f}:=[\hh g\, ;\, f]$, then we obtain a well-defined  weight $w+1$ covector-valued density $ \triangledown^{\scalebox{.7}{$\bm \sigma$}}_a  {\bm f}$ by~\cite{GLW}
$$
 \triangledown^{\scalebox{.7}{$\bm \sigma$}}_a {\bm f}:=[\hh g\, ;\, (\sigma \nabla_a - n_a w)f] \, ,
$$
where $n_a:=\nabla_a \sigma$.
Also,
 if $\bm \omega_a=[\hh g\, ;\, \omega_a]$ is a weight $2-d$ covector-valued density, then its {\it divergence} 
$$
\divergence \bm\omega :=[\hh g\, ;\, \nabla^a \omega_a]
$$
is a well-defined weight $-d$ density.

A weight~$w$ {\it log-density} 
 is also defined by an equivalence class of (metric,function) pairs
 as follows~\cite{GW} 
$${\bm \lambda}:=[\hh g\, ;\, \lambda]=[\Omega^2 g\, ;\, \lambda+w\, \varpi]\, .$$
In particular, given a strictly  positive, weight $w$ density~$\bm \tau$,  we may define its logarithm as the weight $w$ log density
$$\log {\bm \tau}:=[\hh g\, ;\, \log \tau\, ]\, .$$

On occasion it  will be useful to employ the {\it weight operator}~$\w$ defined
acting on the {\it conformal metric} $\bm g_{ab}:=[\hh g\, ; \, g_{ab}]$ and its inverse 
$\bm g^{ab}:=[\hh g\, ; \, g^{ab}]$, a weight $w$ density~$\bm \tau$ and a weight $w$ log-density $\bm \lambda$ by 
$$\w  {\bm g}_{ab}=2{\bm g}_{ab}\, ,\quad
\w  {\bm g}^{ab}=-2{\bm g}^{ab}\, ,\quad
\w  {\bm \tau}=w{\bm \tau}\, ,\quad
\w  {\bm \lambda}=w\, .
$$
Note that the conformal metric and its natural inverse can be employed to perform index contractions for products of tensor densities.

The operator $ \triangledown^{\scalebox{.7}{$\bm \sigma$}}_a$ is well-defined acting on log-densities, for example,
\begin{equation}\label{trianglelog}
 \triangledown^{\scalebox{.7}{$\bm \sigma$}}_a \log {\bm \tau} = [\hh g\, ;\, \sigma \nabla_a \log \tau - n_a]
\end{equation}
is a unit weight density.

It is worth remarking that any dimensionful physical quantity can be regarded as a conformal density, since the transformation $g_{ab}\mapsto \Omega^2 g_{ab}$ amounts to a local choice of unit system while conformal weights then measure physical dimensions
of observables.

\subsection{Defining density}\label{DD}
Given an embedded  hypersurface $\Sigma\subset M$, a {defining density} $\bm \sigma$ is a  weight $w=1$ density ${\bm \sigma}=[\hh g\, ;\, \sigma\, ]$ with zero locus $${\mathcal Z}({\bm \sigma}):=\{P\in M\, |\, \sigma(P)=0\}=\Sigma\, ,$$ and such that $\ext\!\sigma |_P\neq 0$, $\forall P\in \Sigma$ (so the function $\sigma$ is a {\it defining function} for $\Sigma$). For a given hypersurface, a defining density  always exists, at least locally.
 
The {\it ${\bm {\mathcal S}}$-curvature} of a conformal metric~${\bm c}$ and defining density~${\bm \sigma}$
is the weight $w=0$ density ({\it i.e.}, function) defined by
\begin{equation}\label{Scurvy}
{\bm {\mathcal S}}:=\Big[\hh g\, ;\ g^{ab} (\nabla_a \sigma)(\nabla_b \sigma)-\frac{2\sigma}d\,\big(g^{ab}\nabla_a\nabla_b\sigma +\sigma \J\, \big)\Big]\, .
\end{equation}
Working in the scale $g_{ab}$, and denoting $n_a:=\nabla_a\sigma$ and $\rho:=-\frac{1}{d}(\Delta + \J\, )\sigma$, the   ${\bm {\mathcal S}}$-curvature is given by the function
$
n^2+2\rho\sigma
$.

\subsection{The Laplace--Robin operator}\label{LAPROB}
Let $\bm \sigma=[\hh g\, ;\, \sigma]$ be a weight~1 density. Then
the corresponding {\it Laplace--Robin operator}~$\D$ maps weight $w$ scalar conformal densities to 
weight $w-1$~conformal densities according to 
\begin{equation}\label{Ldef}
\D\!{\bm f}:=\big[\hh g\, ;\, (d+2w-2)(\nabla_n+w\rho)f-\sigma
(\Delta+w\J\, )f\, 
\big]\, .
\end{equation}
Note that this a Laplacian-type operator that is degenerate along the zero locus of $\bm \sigma$. In the case that $\bm \sigma$ is a defining density, this restricts to a Robin-type (``Dirichlet plus Neumann'') operator along the corresponding hypersurface~$\Sigma$.

The Laplace--Robin operator  also maps weight~$w$ log-densities to weight $-1$ densities via
\begin{equation}\label{Dlog}
\D\bm \lambda:=
\big[\hh g\, ; (d-2)(\nabla_n \lambda+w\rho) 
 -\sigma (\Delta \lambda+w\J\, )\big]\, .
\end{equation}

The weight and Laplace--Robin operators obey the algebra
$$
[\w,\D]=-\D\, .
$$
The  multiplicative operators 
$
\s\!{\bm f}:={\bm \sigma}{\bm f}
$
and $
\I2\!{\bm f}:={\bm {\mathcal S}}{\bm f}
$,
mapping weight $w$ densities to weight $w-1$ and  $w$ densities  respectively, obey 
$$
[\w,\s]=\s\, ,\qquad \, [\w,\I2]=0\, .
$$
Importantly,  for {\it any} conformal structure and defining density the following algebra holds~\cite{GW}
\begin{equation}\label{algebra}
[\D,\s]=\I2\, \circ \,  (d+2\w)\, .
\end{equation}
Thus, when the ${\bm {\mathcal S}}$-curvature is non-vanishing, the operators
$
x:=\s$, $h:=d+2\w$ and $y:=-\I2^{-1} \D
$
obey the  ${\mathfrak sl}(2)$ Lie algebra
\begin{equation}\label{solgen}
[\x,\y]=\h\, ,\qquad [\h,\x]=2\x\, ,\qquad[\h,\y]=-2\y\, ;
\end{equation}
for reasons linked to its applications, we call this the {\it solution generating algebra}.

The  algebra~\nn{solgen} also holds upon replacing $y:=-\I2^{-1} \circ \D$ by $y:=-\D \circ \I2^{-1}$. The difference between these two choices is encoded by the following lemma:
\begin{lemma}\label{C1}
Suppose the $\bm{\mathcal S}$-curvature is nowhere vanishing, then
acting on densities, the following operator identity holds:
$$[\D, \I2^{-1}]=(\D {\bm {\mathcal S}}^{-1}) 
-2
(\bm \nabla_a {\mathcal S}^{-1}) \, {\bm g}^{ab}
   \triangledown^{\scalebox{.7}{$\bm \sigma$}}_b\, .
$$
\end{lemma}

\begin{proof}
Acting on a weight $w$ density ${\bm f}:=[\hh g\, ;\, f]$ and remembering that ${\bm {\mathcal S}}=[\hh g\, ;\, {\mathcal S}]$ has weight~0,  we have
\begin{equation*}
\begin{split}
[\D, \I2^{-1}] {\bm f}&=\big[\hh g\, ;\, 
\big((d+2w-2)(\nabla_n+w\rho\big)-\sigma (\Delta+w J)\big)
({\mathcal S}^{-1}f)\\
&
\phantom{\big[\hh g\, ;\, 
\big((d+\, \, \, }
-{\mathcal S}^{-1}\big((d+2w-2)(\nabla_n+w\rho\big)-\sigma (\Delta+w J)\big)f
\big]\\[1mm]
&=\big[\hh g\, ;\, 
\big((d-2)(\nabla_n{\mathcal S}^{-1})
-\sigma (\Delta {\mathcal S}^{-1})\big) f
-2\, (\nabla^a {\mathcal S}^{-1})(\sigma \nabla_a-n_aw)f
\big]
\\[1mm]
&=
(\D {\bm {\mathcal S}}^{-1}) {\bm f}
-2
(\bm \nabla_a {\mathcal S}^{-1}) \, {\bm g}^{ab}
   \triangledown^{\scalebox{.7}{$\bm \sigma$}}_a{\bm f}
\, .
\end{split}
\end{equation*}
\end{proof}

The Laplace--Robin operator also enjoys  an integration by parts formula:

\begin{theorem}\label{parts}
Let $\bm f$ and $\bm g$ be densities of weight $1-d-w$ and $w$, respectively. Then~$\D$ is formally self-adjoint and moreover
$$
\bm f \D \bm g - (\D \bm f)\,  \bm g
+\divergence \bm j=0\, ,
$$
where the weight $2-d$ covector-valued density 
$$\bm j_a=\big [\hh g\, ;\, \sigma\big(f\, \nabla_a g-(\nabla_a f)\, g\big)-(d+2w-1)\, n_a fg\big]\, .$$
\end{theorem}

\begin{proof}
The first equality follows simply from writing out  the left hand side of the display  in some scale~$g_{ab}$. 
Thereafter, it remains to 
verify that ${\bm j}_a$ is indeed a density of the quoted weight, which again follows from a direct computation.
\end{proof}

\noindent Because the above result holds for generally curved conformal structures, we expect Theorem~\ref{parts} to be of interest beyond our current context.

\subsection{Conformal hypersurface invariants}

Consider an embedded  hypersurface described by a defining function $\Sigma={\mathcal Z}(\sigma)$.
A {\it hypersurface preinvariant}\ $\mathcal P(g,\sigma)$ amounts to a diffeomorphism invariant quantity built from $\sigma$ and the metric such that 
$$
\mathcal P(g,\sigma)\big|_\Sigma=\mathcal P(g,v\sigma)\big|_\Sigma
$$
for any positive function~$v$ (see~\cite{GW15} for a precise definition). 
A {\it hypersurface invariant} $P(g_{ab},\Sigma)$  is the restriction of a hypersurface preinvariant
to $\Sigma$; per its definition, this depends only on the Riemannian embedding of the hypersurface~$\Sigma$, and in particular  not on the choice of  a defining function. Key examples include the
{\it unit normal}
\begin{equation}\label{unorm}
\hat n_a:=\frac{\nabla_a \sigma}{|\nabla \sigma|}\Big|_\Sigma\, ,
\end{equation}
the {\it first fundamental form}
$$
{ I}_{ab}:=\Big(g_{ab}-
\frac{(\nabla_a \sigma)}{|\nabla \sigma|}\frac{(\nabla_b \sigma)}{|\nabla \sigma|}\Big)\Big|_\Sigma\, ,
$$
the {\it mean curvature}
\begin{equation}\label{H}
H:=\frac{1}{d-1}\, \nabla^a\Big(\frac{\nabla_a \sigma}{|\nabla \sigma|}\Big)\Big|_\Sigma\, ,
\end{equation}
and the {\it second fundamental form} 
$$
\II_{ab}:=\big(\nabla_a
-\frac{(\nabla_a \sigma)}{|\nabla \sigma|}
\frac{(\nabla^c \sigma)}{|\nabla \sigma|}\, \nabla_c\big)
\Big(\frac{\nabla_b \sigma}{|\nabla \sigma|}\Big)\Big|_\Sigma\, .
$$
Hypersurface invariants obey various non-trivial identities, the most of important of which include  the identification of the intrinsic hypersurface metric~$\bar g_{ab}$ with the first fundamental form, and the Gau\ss\ equation expressing the difference between ambient and hypersurface curvatures in terms of the second fundamental form:
\begin{equation}\label{Gauss}
{ I}_{ab}=\bar g_{ab}\, , \qquad
R^\top_{abcd}|_\Sigma=\bar R_{abcd}-2\II_{a[c}\II_{d]b}\, .
\end{equation}
Here and throughout, we  use a superscript~$\top$ to denote orthogonal projection
onto hypersurface-tangential directions.
Note that $I_{ab}^\top=I_{ab}$  and $\II_{ab}^\top =\II_{ab}$. Indeed, using that the projected tangent bundle $TM^\top|_\Sigma$ and the hypersurface tangent bundle $T\Sigma$ are isomorphic, we may use the same indices to label host space and hypersurface tensors.

We will need the following technical result for the mean curvature:
\begin{lemma}\label{Hlemma}
Let $(g,\sigma)$ be a metric and a defining function for a hypersurface~$\Sigma$ such that the corresponding
 {\it ${\bm {\mathcal S}}$-curvature} 
 obeys
 \begin{equation}\label{S1}
 {\bm {\mathcal S}}=\big[\hh g\, ;\, 1+\mathcal{O}(\sigma^2)\big]\, .
 \end{equation}
 Then along $\Sigma$
 $$
 \rho:=-\, \frac{\Delta\hspace{-2.8mm} \phantom{\J}^{\, g}\,  \sigma + \J^{\, g}\sigma}{d}\stackrel\Sigma= -H\, .
 $$
\end{lemma}

\begin{proof}
This result was originally obtained in~\cite[Section 3.1]{Goal} for the case ${\bm {\mathcal S}}=1$ and the proof proceeds along similar lines to that given there. Starting with the preinvariant  on the right hand side of Equation~\nn{H} we have
$$
\nabla^a\Big(\frac{\nabla_a \sigma}{|\nabla \sigma|}\Big)
=\frac{\Delta \sigma}{|\nabla \sigma|}
\ +\ 
\nabla_n\,  (|\nabla \sigma|^{-1})\, ,
$$
where $n_a:=\nabla_a\sigma$. Comparing Equations~\nn{Scurvy} and~\nn{S1} yields $|\nabla \sigma|^2 + 2\rho \sigma=1+{\mathcal O}(\sigma^2)$, so that along~$\Sigma$ it follows that $|\nabla\sigma|=1$, $\Delta \sigma=-d\rho$  and
$$
\nabla_n\,  (|\nabla \sigma|^{-1})\stackrel\Sigma=\rho\, .
$$
Thus
$$
\nabla^a\Big(\frac{\nabla_a \sigma}{|\nabla \sigma|}\Big)\stackrel\Sigma=-(d-1)\rho\, .$$
\end{proof}

\medskip

When $P(\Omega^2g,\Sigma)=\Omega^w P(g,\Sigma)$, the equivalence class of hypersurface invariants 
$$
{\bm P}:=[\hh g\, ;\, P(g,\Sigma)]=[\Omega^2g\, ;\, \Omega^{w}P(g,\Sigma)]
$$
defines a {\it conformal hypersurface invariant}. Important standard examples include the 
weight $w=1$~{\it unit normal density} and weight $w=2$ {\it first fundamental form density}
$$
{\bm {\hat n}}_a:=[\hh g\, ;\, \hat n_a]\ \mbox{ and } \ 
{\bm {I}}_{ab}:=[\hh g\, ;\, I_{ab}]\, ,
$$
as well as the  (weight $w=1$) {\it trace-free second fundamental form density}
$$
{\bm {\IIo}}_{ab}:=[\hh g\, ;\, \II_{ab}-H\,  I_{ab}]\, .$$
We define the 
weight $w=-2$ density
$$
\bm K:=\, \bm{\IIo}_{ab}\, \bm{\IIo}^{ab}\, .
$$
For rigid surfaces, this gives a
 measure of the energy density due to bending. It also appears as the Lagrangian density for a rigid string~\cite{Polyakov}; hence we call $\bm K$ the {\it rigidity density}. As a simple consequence of the Gau\ss\  Equation~\nn{Gauss}, in ambient dimension $d\geq 3$, the rigidity density can be reexpressed in terms of  Riemann and mean curvatures:
\begin{equation}\label{K2J}
\bm K=(d-2)\big[\hh g\, ;\, 
2\big(\J-\Rho_{ab}\hat n^a\hat n^b
-\, \bar{\!\! \J}\, \big)+(d-1)H^2\big]\, .
\end{equation}

\medskip

We shall also need the weight $w=0$ {\it Fialkow tensor}
defined in dimensions $d> 3$ by~\cite{Grant,Stafford}
\begin{eqnarray*}
\bm {\mathcal F}_{ab}&:=&\Big[\, g\, ; \, \Rho^\top_{ab}-\bar\Rho_{ab}+H\, \IIo_{ab}+\frac12\,  \bar g_{ab} H^2\Big]\\[1mm]
&=&\frac1{d-3}\Big(\bm \IIo_{\!a}^c\, \bm \IIo_{cb}-\frac1{2(d-2)} \bm{  I}_{ab}\,  \bm{K} - \bm W_{\!cabd}\, \bm{\hat n}^c\bm{\hat n}^d\Big)
\, .
\end{eqnarray*}
The second line above follows from a standard application of the Gau\ss\  equations (see \cite{YuriThesis,GW15}); we  have used conformal invariance of the Weyl tensor $W_{ab}{}^c{}_d$ to define the weight $2$ density~$\bm W_{\! abcd}:=[\hh g\, ;\, W_{abcd}]$. Finally, in dimension~$d=4$, the {\it hypersurface Bach tensor density} of weight~$-1$ is defined by~\cite{GGHW15}
$$
\bm B_{ab}:=\big[\hh g\, ;\, \big(\hat n^c C_{c(ab)}\big)^{\!\top}+H W_{cabd}\,  \hat n^c \hat n^d -\nablab^c\big((\hat n^d W_{d(ab)c})^{\!\top}\big)\big]\, .
$$
In the above, $C_{abc}$ is the ambient Cotton tensor. 
Continued to dimensions greater than four, for almost Einstein structures, the first term on the right hand side above is linked to the ambient Bach tensor~\cite{Goal,GLW}.

\subsection{Extrinsic conformal Laplacian powers and BGG operators}

Given a hypersurface $\Sigma$ and a corresponding defining density $\bm \sigma$,
a smooth operator~$\Operator$, whose domain is densities on $M$, 
is said to be {\it tangential}  if 
$$\Operator\circ \s = \s\circ \, \widetilde \Operator\, ,$$
for some other smooth operator $\widetilde\Operator$.
Tangential operators are useful since they can be used to define and efficiently treat operators on hypersurface densities~$\bm{\bar f}$ via
$$
\ \ \overline{\Smidge\Operator\smidge}\,  \bm{\bar f} := \big(\!\Operator \!\bm f\big)\big|_\Sigma\, ,
$$
where $\bm f$ is {\it any} smooth extension of $\bm {\bar f}$ to $M$.

\medskip

A key point for us is that
nontrivial tangential operators can be constructed using the solution generating algebra~\nn{solgen} 
 by employing the standard $\mathfrak{sl}(2)$ enveloping algebra identity
$$
[y^k,x]=-ky^{k-1}(h-k+1)\, .
$$
This implies that the operator
\begin{equation}\label{tangID}
\GJMS_k^{\scalebox{.7}{$\bm\sigma$}}:= (-\I2^{-1} \D)^k
\end{equation}
is tangential when acting on densities of weight $\frac{k-d+1}{2}$.
In general this operator depends on the choice of defining density $\bm \sigma$. However, in Section~\ref{singYam} we present a canonical defining density~$\bm{\bar\sigma}$ obtained by solving a singular version of the Yamabe problem,
this yields {\it extrinsic conformal Laplacian powers} 
$$\GJMS_k:=\overline{\GJMS\hspace{-.13mm}}{\,}_k^{\scalebox{.8}{$\bm{\bar\sigma}$}}\, ,$$ 
determined entirely by the data $(M,\bm c,\Sigma)$
(for orders $k\geq d$ the above definition must be slightly modified, see~\cite{GW15} for details). The simplest example is when $k=2$. In this case $P_2$ is an  extrinsic generalization of the hypersurface Yamabe operator
$$
\GJMS_2\,  [\hh g\, ;\, \bar f\, ]=\Big[\hh g\, ;\, \Big\{\bar \Delta+\Big(1-\frac\db 2\Big)\Big(\ \bar{\!\!\J} \ -\frac{K}{2(d-2)}\  \Big) \Big\} \ \bar f \, 
\Big]\, .
$$
Here $K:=\IIo_{ab}\, \IIo^{ab}$ is the rigidity density.
For $k$ even, the operators $\GJMS_k$ have leading term proportional to the Laplacian power $\bar\Delta^{\frac k2}$, and are therefore extrinsic analogs of GJMS operators.

A second class of non-trivial hypersurface operators is linked to the BGG construction of~\cite{CSS}.
The very general BGG technology provides  sequences of conformally invariant  operators associated to finite dimensional irreducible 
representations of the conformal group.
Specializing to hypersurfaces, the {\it first BGG operator} associated to the defining (or vector) representation acts on weight one densities and therefore also conformal hypersurface invariants $\bm {\bar f}$ according to  
$$
\D_{ab} \bm {\bar f}=[\hh\bar g\, ; \, (\nablab_{(a}\nablab_{b)\circ} + \bar \Rho_{(ab)\circ}) \bar f\, ] \, ,\quad d \geq 4\, .
$$
In hypersurface dimension two, the above (intrinsically defined) operator is unavailable. However, in that case, there exists an  {\it extrinsic hypersurface BGG operator}~\cite{GGHW15}. We will need the formal adjoint of this operator 
which maps 
rank~2, weight $-3$ symmetric, trace-free, conformal hypersurface tensor densities $\bm X^{ab}:=[\bar g\, ;\, X^{ab}]$ to a
conformal hypersurface density of weight $-3$ according to
\begin{equation}\label{dualBGG}
\D_{ab}^* \bm X^{ab} = [\hh g\, ;\, \nablab_a \nablab_b  X^{ab}+\Rho_{ab}X^{ab}+H\, \IIo_{ab}X^{ab} ]\, . 
\end{equation}

\subsection{Integrated densities}\label{integrated}

Recall that a weight $-d$ density~$\bm f=[\hh g\, ;\, f]$ can be invariantly integrated over a conformal $d$-manifold $M$ or some region $D\subset M$ since the volume element~$\ext \! V^g$  of $g_{ab}\in {\bm c}$  
defines a weight $d$, measure-valued density
$$
\bm{\ext \! V}:=[\hh g\, ;\, \ext \! V^g]\, ,
$$
because $\ext \! V^{\Omega^2g}=\Omega^{d} \ext \! V^g$. 
Hence, we may define the conformally invariant integral over~$\bm f$~by
$$
\int_D\bm f \, :=\,  \int_D \ext \! V^g \, f\, .
$$
Similarly, for hypersurface conformal invariants, the induced metric $\bar g_{ab}=I_{ab}$ defines an ``area'' element $\ext \! A^{\bar g}$ ({\it i.e.}  the volume form of $\bar g$ along the hypersurface~$\Sigma$). From this we may build the  
weight $d-1$ density $\bm{\ext \! A}:=[\hspace{.5mm}\bar g\, ;\, \ext \! A^{\bar g}]$. Thus, for any weight $1-d$, scalar,  conformal hypersurface invariant $\bm P:=[\hh g\, ;\, P(g,\Sigma)]$  we define
$$
\int_\Sigma \bm P\, :=\, \int_\Sigma \ext \! A^{\bar g} \, P\, .  
$$ 

\subsection{The Dirac-delta density}\label{Dirac}

We now describe of the main ideas
of our approach:
We will employ the Dirac delta function to express hypersurface integrals as bulk integrals.
  
Given a defining function~$s$ for a hypersurface~$\Sigma$ 
and 
$\bar f := f|_\Sigma$ with $f\in C^\infty M$, we may then  rewrite the integral of $\bar f$ as a bulk integral according to (see, for example~\cite{GGHW15} or~\cite{Osher})
\begin{equation}
\label{deltasurface}
\int_\Sigma \ext \! A^{\bar g}\,  \bar f \, = \, 
\int_{\, \widetilde{\!D}} \ext \! V^g\, \delta(s)\,  |\nabla s|\,  f\, ,\end{equation}
where $\, \widetilde{\!D} \supset  {\rm supp}(\bar f)\subset \Sigma$ is some region in $M$ that includes the support of $\bar f$. 

Given a metric $g$ the function $f$  determines a weight $1-d$ density $[\hh g\, ;\, f]=:\bm f$, and the above display can be expressed as an integral over densities. This is particularly important for us when the hypersurface is given in terms of a defining density~$\bm \sigma=[\hh g\, ;\, \sigma]$. 
Then we may use the
the distributional identity (valid for non-vanishing $\Omega$; see Section~\ref{distributions} below)
$$
\delta(\Omega\sigma)=\Omega^{-1}\delta(\sigma)
$$
to infer that
$$
\bm\delta:=[\hh g\, ;\, \delta(\sigma)]
$$
is a weight $w=-1$ (distribution-valued) density. Since, in a scale $g_{ab}$, we have that~$\sigma$ is a defining function, it follows that  the $\bm {\mathcal S}$-curvature of~$\bm \sigma$ 
obeys
$$
{\mathcal S} = |\nabla \sigma|^2\quad \mbox{ along }\Sigma\, .
$$
Hence
\begin{equation}\label{deltaintegral}
\int_{\, \widetilde{\!D}} \bm \delta \, \sqrt{\bm{\mathcal S}} \, \bm f
= \int_{\, \widetilde{\!D}} \ext \! V^g \, \delta(\sigma)\, \sqrt{\mathcal S}\, f = \int_\Sigma \bm {\bar f}\, , 
\end{equation}
where $\bm{\bar f}=\big[\bar g_{ab}\, ;\, f|_\Sigma\big]$. We will often drop the bar notation when using this formula.
This identity allows efficient handling of integrated conformal hypersurface invariants. Note that this does not require using an  extension $\bm f$ of $\bm{\bar f}$ which is a  hypersurface preinvariant, but for variational problems it will be useful to do~so. 

\subsection{Distributional identities}\label{distributions}

Standard distributional identities (on ${\mathbb R}$) for the Dirac delta and Heaviside step function such as
$$
\theta'(x)=\delta(x)\, ,\quad\!
x\delta(x)=0\, ,\quad\!
x\delta'(x)=-\delta(x)\ \mbox{ and }\ 
x\delta^{(n)}(x)=-n\delta^{(n-1)}(x)\, ,\!\quad n\in {\mathbb Z}_{\geq 1}\, , 
$$
and their consequences will play a crucial role in our derivation of volume anomalies and divergences. Such identities hold when integrating against suitable test functions. Some care is required to justify their use, but the  details are essentially the same in each case.
Therefore we explain the key ideas here and suppress the details when presenting the computations below.

We wish to apply distributional identities  to the situation where the variable~$x\in {\mathbb R}$ is replaced by a defining function~$\sigma$ for a  hypersurface $\Sigma$ embedded in a manifold $M$; in particular we will be dealing with the distribution $\delta(\sigma)$ and derivatives thereof.
In our computations we assume that the hypersurface~$\Sigma$ is closed (compact without boundary) and that  in a neighborhood of~$\Sigma$ the bulk manifold~$M$ is a product $\Sigma\times I\subset M$
 where $I$ is some small open interval about $0$. Moreover, we 
 assume that the defining function $\sigma$ pulls back to the standard coordinate $x$ on $I$.
 In particular, in what follows, we assume that bulk integrals are over regions contained in $\Sigma\times I$ and so can be treated by Fubini's theorem.
 
 Then to treat  distributional computations in detail, 
 we introduce  a fixed, smooth,  cutoff function
 $\chi$ taking the value 1 on 
 the neighborhood $\Sigma\times I'$, for some open  interval~$I'\subset I$.
 Thus, integrals involving the distributions $\theta(\sigma)$ or $\delta(\sigma)$ and their derivatives are  defined by the expressions given below but with the  
 insertion of  the test  function $\chi$.
 It is then easily verified that these integrals have their intended meaning and we leave the details of
 the distributional calculations to the reader.
 
The results we obtain this way are local terms integrated along the hypersurface~$\Sigma$. Hence, they apply
beyond the situation where~$\Sigma$ 
is closed, to more general settings as depicted in Diagram~\nn{iamnotclosed} and applied in the  example computation  given in Section~\ref{Kasner}.

\section{Renormalized volume}\label{RV}

\subsection{Conformal infinity}

Let~$(M,{\bm c},\Sigma)$ denote a conformal manifold $(M,\bm c)$ equipped with an 
 embedded, oriented, hypersurface or boundary component~$\Sigma$.
 Given this data and  some choice of defining density~${\bm \sigma}$ for~$\Sigma$ (see Section~\ref{DD}),  then  on the manifold~$\widehat{ \!M}:=M\backslash \Sigma$ we may  extract a canonical metric $g^o$ such that on one side of $\Sigma$
$$
{\bm \sigma}=[\hh g^o\, ;\, 1]\, .
$$
The metric $g^o$ is  then singular along~$\Sigma$ and the hypersurface~$\Sigma$ is a {conformal infinity} of~$g^o$. 
 The metric $g^o$ may be used to compute volumes of bounded domains $\, \widehat{\! D}\, \subset \,\widehat{\! M}$ via
$$
\Vol(\, \widehat{\! D};{\bm \sigma})=\int_{\, \widehat{\!D}} \ext \! V^{g^o}\, ,
$$
where $ \ext \! V^{g^o}$ is the volume form of the metric $g^o$. Rewriting the above display in terms of a general equivalence class representative $[\hh g\, ;\, \sigma]$ we have
\begin{equation}\label{VolDhat}
\Vol(\, \widehat{\! D};{\bm \sigma})=\int_{\, \widehat{\!D}} \, \frac{\ext \! V^{g}}{\sigma^d}\ =\ 
\int_{\, \widehat{\!D}} \, \frac{1}{\bm \sigma^d}
 \, ,
 \end{equation}
which at the same time manifests the conformal invariance  of $\Vol(\, \widehat{\! D}\, ;\hh {\bm \sigma})$ (as a functional of $(\bm c,\bm \sigma)$) while emphasizing that it would be singular for regions intersecting the  hypersurface~$\Sigma$.

\subsection{The regulated volume}
\label{regvol_sect}

We now wish to study bounded regions $D$  for which the
intersection  $\partial D\cap \Sigma$ is non-vanishing
and admits a finite collar neighborhood contained in~$D$,
as depicted in the first diagram below.
In that case the analog of the expression~\nn{VolDhat} is divergent. 
  Therefore, working on the side of~$\Sigma$ where $\sigma$ is positive, we regulate this expression by inserting a cut-off 
$$
\theta(\sigma/\tau-\varepsilon)\, ,
$$
where $\theta:{\mathbb R}\to\{0,1\}$ is the Heaviside step function (with support  ${\mathbb R}_{\geq 0}$) and $\bm\tau=[\hh g\, ;\, \tau]$ is any true scale. The freedom to choose different regulators is captured 
by the choice of the true scale $\bm\tau$. Given $\bm \tau$, we define the corresponding  {\it regulated volume}~$\Vol_\varepsilon$ by
\begin{equation}\label{vregd}
\Vol_\varepsilon(D,\Sigma):=
\int_D\, \frac{\bm \theta_\varepsilon}{\bm {\sigma}^d}=
\int_D \ext \! V^g\, \frac{ \theta(\sigma/\tau-\varepsilon)}{{\sigma}^d}\, .
\end{equation}
Here we have used the weight~0 density  $\bm \theta_\varepsilon:=[\hh g\, ;\, \theta(\sigma/\tau-\varepsilon)]$.
By construction this definition agrees with Definition~\ref{def1} with $\Sigma_\varepsilon$ determined by the zero locus of the function $\bm \sigma/\bm \tau-\varepsilon$.
The above integral computes the volume of the darker shaded region $D_{\varepsilon}$ depicted in the second diagram displayed below:
\begin{equation}\label{iamnotclosed}\end{equation}
\vspace{-2.5cm}
\begin{center}
\includegraphics[scale=.24]{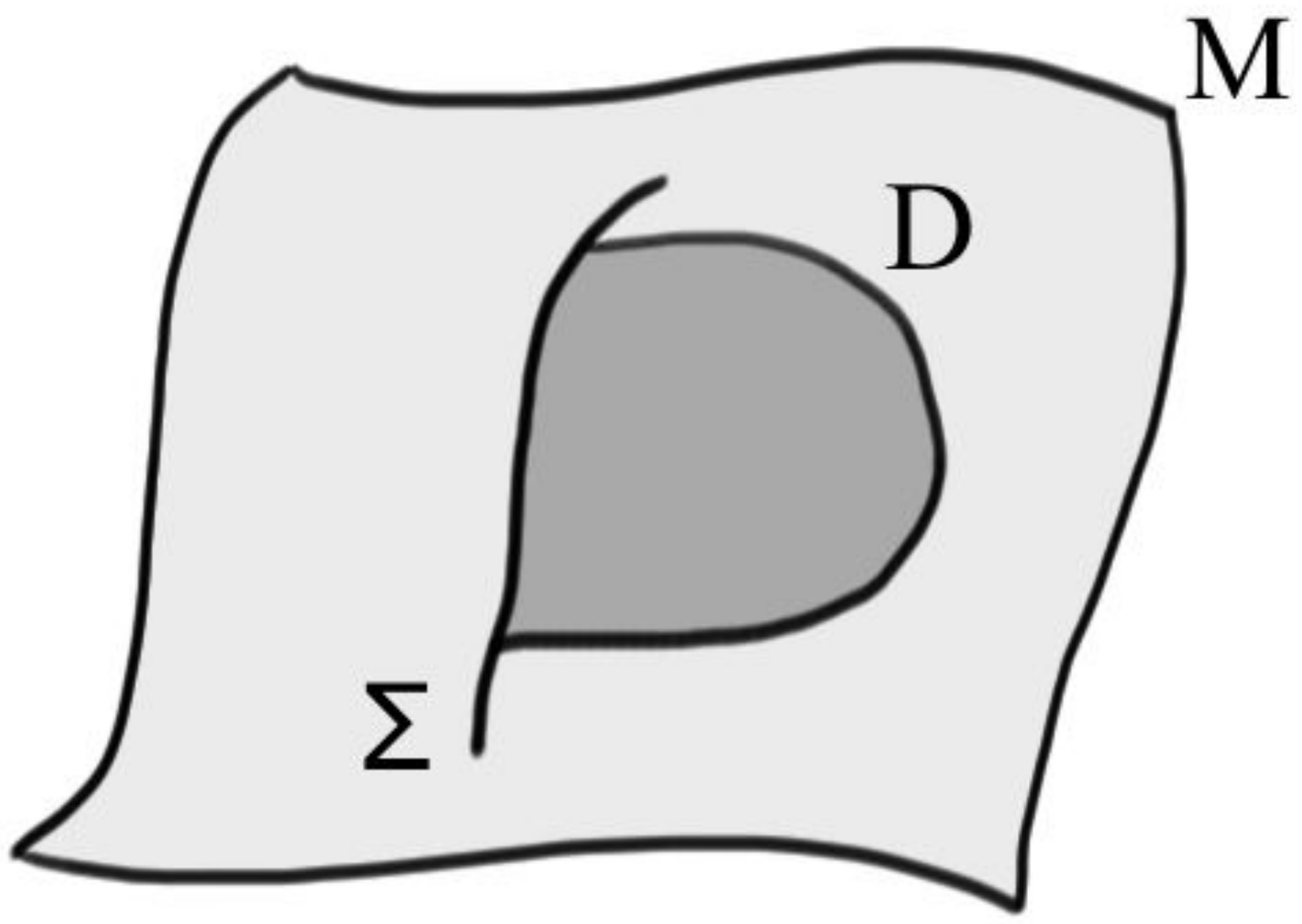}
\qquad
\includegraphics[scale=.24]{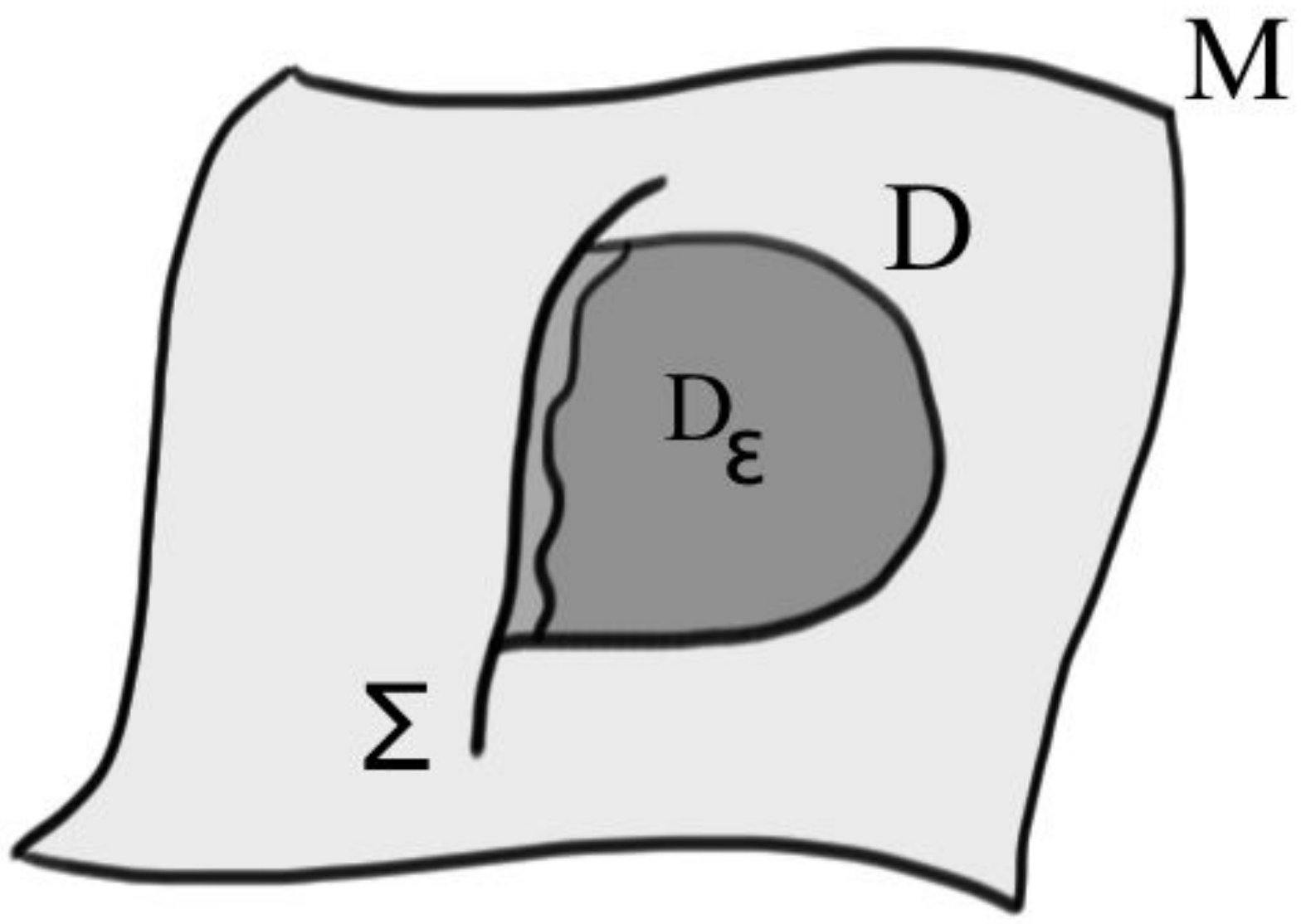}
\end{center}
\vspace{-1cm}

A technical remark will be important when dealing with surface terms  in Section~\ref{HolFor}: The regulated volume is unchanged if we extend the region of integration $D$ beyond the hypersurface $\Sigma$ to a new, compact, region~$\, \widetilde{\!D}$ as depicted below. We assume this is always possible; for the case $\Sigma=\partial M$ we choose an extension to enable this.
\vspace{-2.4cm}
\begin{center}
\includegraphics[scale=.32]{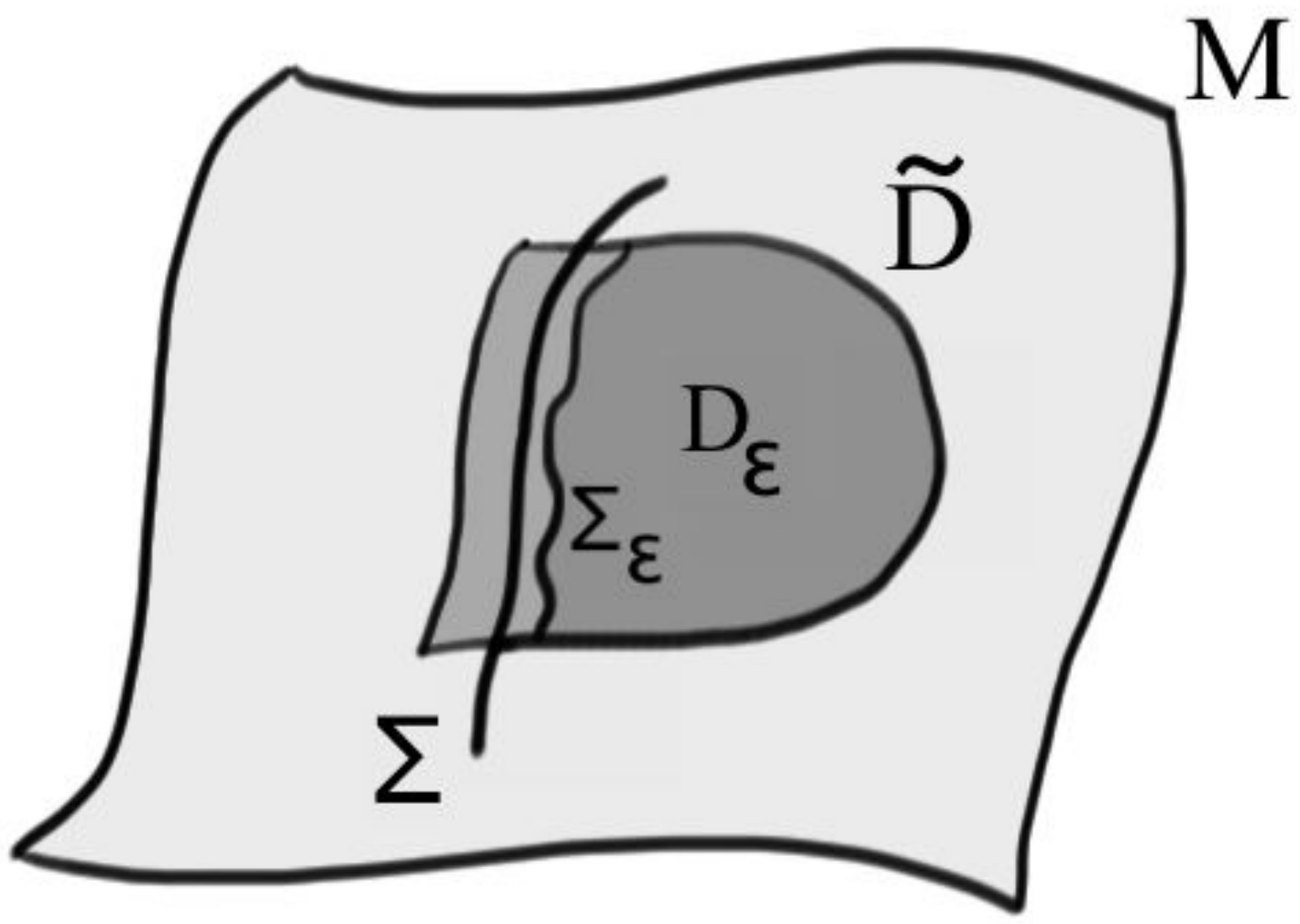}
\end{center}
\vspace{-2.3cm}
 Since we are ultimately interested in the dependence of the regulated volume on the hypersurface embedding, in the following we will  write  $\Sigma$ for the intersection~$\Sigma\cap D$. 
Alternatively, one can consider the conformally compact setting common in applications where $\Sigma=\partial M$ and $D=M$. In the case where $M$ has a puct structure and $\Sigma$ is compact as discussed in Section~\ref{distributions}, the last diagram above is replaced by:

\vspace{-5.0cm}
\begin{center}
\includegraphics[scale=.53]{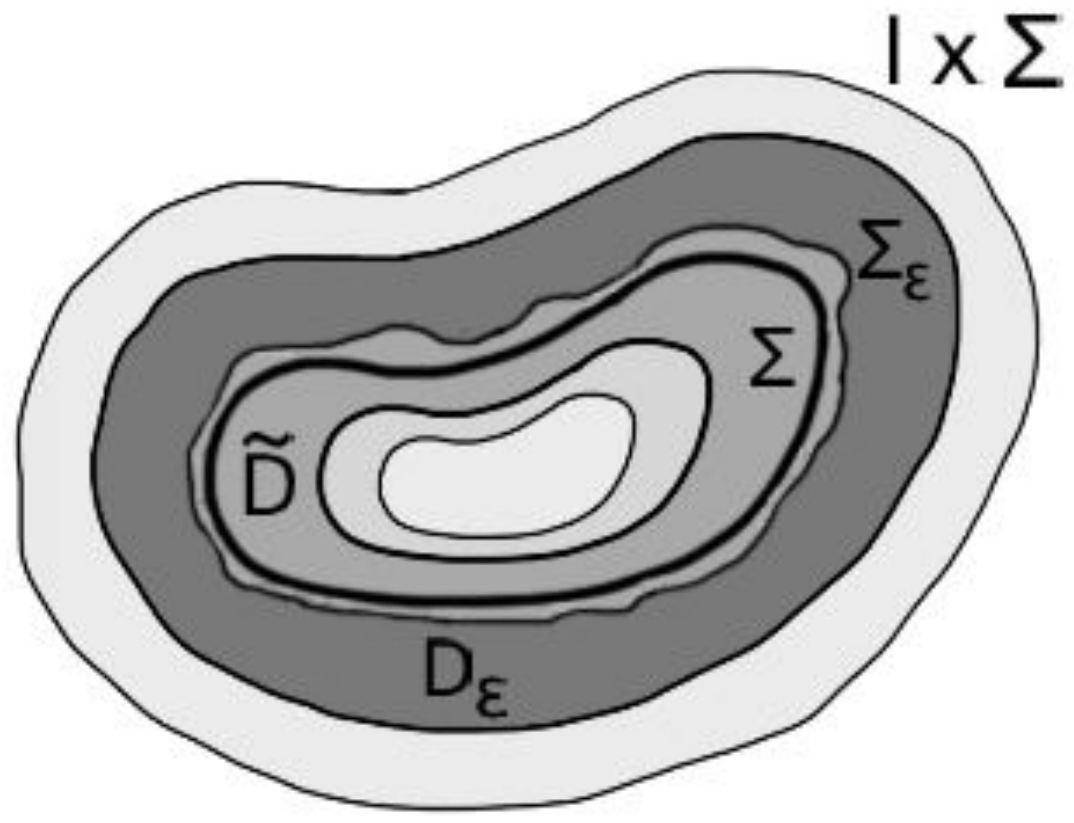}
\end{center}
\vspace{-4.9cm}

 \noindent
 In all cases, the
 regulated volume is given by
\begin{equation}\label{regvol}
\Vol_\varepsilon(D,\Sigma):=
\int_{\, \widetilde{\!D}}\, \frac{\bm \theta_\varepsilon}{\bm{\sigma}^d}\, \ .
\end{equation}

\smallskip

\subsection{The $\varepsilon$ expansion}

Our strategy will be to show that  the  regulated volume is a Laurent series plus a logarithm in~$\varepsilon$. Except for the 
constant term, the coefficient of each term will be  a hypersurface integral over $\Sigma$.
For our purposes the standard  distributional identity 
$$
\frac{d\theta(\sigma/\tau-\varepsilon)}{d\varepsilon}=-\delta(\sigma/\tau-\varepsilon)
$$
is key to studying the analyticity properties of the regulated volume $\Vol_\varepsilon$ as a function of $\varepsilon$. By the meaning of this identity this implies
$$
\frac{d\Vol_{\varepsilon}}{d\varepsilon}=-\int_{\, \widetilde{\!D}} \frac{\ext \! V^g}{\sigma^d}\, \delta(\sigma/\tau-\varepsilon)=-\, \varepsilon^{-d}\int_{\, \widetilde{\!D}} \frac{\ext \! V^g}{\tau^d}\, \delta(\sigma/\tau-\varepsilon)\, .
$$
We now need to analyze the integral $$I(\varepsilon):=\int_{\, \widetilde{\!D}} \frac{\ext \! V^g}{\tau^d}\, \delta( \sigma/\tau-\varepsilon)\, .$$ To that end, consider the function $ s:=\sigma/\tau$. Since $f=\tau^{-d}/|\nabla  s|$ is  smooth in a neighborhood including~$\Sigma$, we may rewrite this expression as a hypersurface integral by employing the delta function 
identity~\nn{deltasurface}:
$$I(\varepsilon)=\int_{\Sigma_{\varepsilon}} \ext \! A^{\bar g_\varepsilon}\,  (\tau^{-d}/|\nabla  s|)|_{\Sigma_\varepsilon}\, .$$
 We have assumed $D$ such that $\Sigma_{\varepsilon}$ is bounded. Since all functions in the integral are smooth,  the hypersurface integral $I(\varepsilon)$ depends smoothly on $\varepsilon$ and, for small enough $\varepsilon>0$, 
 may be written as a Taylor series with error term.
Hence it follows that the regulated volume is the sum of Laurent series terms about $\varepsilon=0$, plus a $\log$ term:
\begin{equation}\label{vs}
\Vol_\varepsilon=\sum_{\scalebox{.7}{$\scriptstyle k\, \in\, \big\{d-1,\ldots, 1
\big\}$}}\frac{v_k}{\varepsilon^k}\  +\ \Vol_{\rm ren}\  + \   {\mathcal A} \, \log\varepsilon+\varepsilon \, \mathcal{R}(\varepsilon)\, ,
\end{equation}
where $\mathcal{R}(\varepsilon)$ is smooth.
The $\varepsilon$ independent part of this series $\Vol_{\rm ren}=:v_0$ defines  the renormalized volume and the $\log \varepsilon$ coefficient~${\mathcal A}$ is the anomaly. Computing ${\mathcal A}$ in full generality and understanding its link to 
extrinsically coupled $Q$-curvatures
is a main goal of our work.

\subsection{Expansion coefficients}
\label{expcoeffs}

To extract the anomaly we 
employ the formula
$$
{\mathcal A}=
\frac{1}{(d-1)!}\, 
\frac{d^{d-1}}{d\varepsilon^{d-1}}
\Big(\varepsilon^d\, \frac{d \Vol_\varepsilon}{d\varepsilon}
\Big)\Big|_{\varepsilon=0}
=-\frac{1}{(d-1)!}\, 
\frac{d^{d-1}I(\varepsilon)}{d\varepsilon^{d-1}} \Big|_{\varepsilon=0}\, ,
$$
where 
$$
I(\varepsilon)=\int_{\, \widetilde{\!D}}\frac{\ext \! V^g}{\tau^{d-1}}\, \delta(\sigma-\varepsilon\tau) \, .
$$
This gives a  simple formula for the anomaly
\begin{equation}\label{preQ}
{\mathcal A}=\frac{(-1)^d}{(d-1)!}\, \int_{\, \widetilde{\!D}}
\, \bm \delta^{(d-1)}\, .
\end{equation}
Here $\delta^{(k)}(x):=\frac{d^k\delta(x)}{dx^k}$ and $$\bm \delta^{(k)}:=[\hh g\, ;\, \delta^{(k)}(\sigma)]$$ is a weight $-k-1$ distribution-valued density. Importantly, Equation~\nn{preQ} shows that the anomaly~${\mathcal A}$ is independent of the choice of regulating scale~$\bm \tau$.

It is also not difficult to generate similar formul\ae\ for the coefficients $v_{k\neq 0}$ (and $k\leq d-1$) by noting 
$
v_{k\neq 0}=\frac{1}{(d-1-k)!\, k}
\frac{d^{d-1-k}I(\varepsilon)}{d\varepsilon^{d-1-k}}\Big|_{\varepsilon=0}
$
so that
\begin{equation}\label{coeffs}
v_{k\neq 0}=\frac{(-1)^{d-k-1}}{(d-1-k)!\, k}\, \int_{\, \widetilde{\!D}} \frac{\bm \delta^{(d-1-k)}}{\bm \tau^k}
\, .
\end{equation}
As expected, these coefficients do depend on the regulating scale~$\bm \tau$. We gather together the results  established above in the following theorem:

\begin{theorem}\label{distdivsan}
The regulated volume $\Vol_\varepsilon:=\Vol_\varepsilon(D,\Sigma)$ as defined in Equation~\nn{regvol}
depends on $\varepsilon$ according to
\begin{equation*}
\begin{split}
\Vol_\varepsilon & \ =\ 
\frac1{d-1}\frac1{\varepsilon^{d-1}}\int_{\, \widetilde{\!D}}
\frac{\bm \delta}{{\bm \tau}^{d-1}}
\ -\ \frac{1}{d-2}\frac1{\varepsilon^{d-2}}
\int_{\, \widetilde{\!D}}
\frac{\bm \delta'}{{\bm \tau}^{d-2}}
\ +\ 
\frac{1}{2(d-3)}\frac1{\varepsilon^{d-3}}
\int_{\, \widetilde{\!D}}
\frac{\bm \delta''}{{\bm \tau}^{d-3}}\ +\ 
\\[2mm]
&\ \  \cdots \ +\ 
\frac{\ (-1)^{d-2}\ }{(d-2)!}\frac1{\varepsilon}
\int_{\, \widetilde{\!D}}
\frac{\bm \delta^{(d-2)}}{{\bm \tau}}
\ + \
\frac{(-1)^d}{(d-1)!}\,\log\varepsilon\  \int_{\, \widetilde{\!D}}
\, \bm \delta^{(d-1)}
\ +\ \Vol_{\rm ren}\ +\  \varepsilon\,  {\mathcal R}\, (\varepsilon)\, ,
\end{split}
\end{equation*}
where the renormalized volume $\Vol_{\rm ren}$ is independent of~$\varepsilon$ and  ${\mathcal R}(\varepsilon)$ is smooth.

\end{theorem}

\subsection{Holographic formul\ae}\label{HolFor}


The following technical result for powers of
the Laplace--Robin operator acting 
on $\bm \delta$ is the key tool for generating  a holographic formula for the anomaly ${\mathcal A}$.

\begin{proposition}\label{delta_deriv}
Let ${\mathbb Z}_{>0}\ni j\leq d-1$ and suppose the $\bm{\mathcal S}$-curvature is nowhere vanishing. Then 
$$
(\I2^{-1}\!\D)^j \, \bm \delta= (d-j-1)\cdots(d-3)(d-2)\, \bm \delta^{(j)}\, .
$$
\end{proposition}
\begin{proof}
The proof is by induction. Consider first the base case $j=1$. We choose some scale $g$ and then compute
$$
\D \delta( \sigma)=(d-4)(\nabla_n-\rho)\delta(\sigma)-\sigma(\Delta - \J\,  )\delta(\sigma)\, .
$$
Using the distributional identities (see Section~\ref{distributions}) and chain rule  we have
$$x\, \delta(x)=0\, ,\quad x\, \delta'(x)=-\delta(x)\, ,$$ and   $\nabla_n\delta(\sigma)=n^a (\nabla_a \sigma) \delta'(\sigma)=n^2 \delta'(\sigma) = ({\mathcal S}-2\rho\sigma)\delta'(\sigma)$, so that (suppressing the~$\sigma$ dependence of the delta functions)
$$
\D\delta=(d-4){\mathcal S}\, \delta'+(d-4)\rho\, \delta+[\Delta,\sigma]\, \delta\, .
$$
But $[\Delta,\sigma]=2\nabla_n+(\nabla.n)$ and $\nabla.n=-d\rho-\J\sigma$ so 
$
\D\delta=(d-2){\mathcal S}\, \delta'
$ whence
$$
\I2^{-1}\!\D\, \bm\delta = (d-2)\, \bm \delta'\, .
$$ 
\vspace{-1mm}

For the induction step we use the further identity 
\begin{equation}\label{xdj}
x\, \delta^{(j)}(x)=-j\, \delta^{(j-1)}(x)\, ,\quad j\in {\mathbb Z}_{\geq 1}\, ,
\end{equation} 
to compute (again in some choice of scale)
\begin{eqnarray*}
\D\delta^{(j-1)}&=&
(d-2j-2)(\nabla_n-j\rho)\, \delta^{(j-1)}-\sigma(\Delta - j\J\, )\, \delta^{(j-1)}\\[1mm]
&=&
(d-2j)\big(\nabla_n-(j+1)\rho\big)\delta^{(j-1)}+(j-1)\big(\Delta-(j-1)\J\, \big)\delta^{(j-2)}
\, .
\end{eqnarray*}
Now
$\Delta\, \delta^{(j-2)}
=(\nabla.n)\, \delta^{(j-1)}
+n^2\, \delta^{(j)}=\mathcal S\, \delta^{(j)}
-(d-2j)\rho\, \delta^{(j-1)}
+(j-1)\J\, \delta^{(j-2)}
$
and
 $\nabla_n\delta^{(j-1)}=(\mathcal S-2\rho\sigma)\,  \delta^{(j)}=\mathcal S \, \delta^{(j)}
+2j\rho\,\delta^{(j-1)}
$. Thus~$
\D\delta^{(j-1)}=
(d-j-1)\, \mathcal S \, \delta^{(j)}
$, whence
\begin{equation}\label{recursion}
\I2^{-1}\!\D\bm\delta^{(j-1)}=(d-j-1)\bm\delta^{(j)}\, .
\end{equation}
\end{proof}


\noindent
We shall need the following related result:

\begin{proposition}\label{Lfdelta}
Let $\bm f$ be a weight zero density. Then
$$
\D (\bm f\, \bm \delta^{(d-2)}) = (\D \bm f) \bm \delta^{(d-2)}\, .
$$
In particular
$\D \bm \delta^{(d-2)}=0\, .$
\end{proposition}

\begin{proof}
Were the operator $\D$ to obey the Leibniz rule, the result would be a direct consequence of Equation~\nn{recursion} for $j=d-1$. Thus it suffices to verify that the non-Leibniz terms in $\D (\bm f\, \bm \delta^{(d-2)})$ vanish. This is a straightforward computation that requires only the methods used in the proof of the preceding lemma. 
\end{proof}

\subsubsection{Divergences}

We now apply the above  Proposition~\ref{delta_deriv} and the formal self-adjoint property of the Laplace--Robin operator given in Theorem~\ref{parts}
to translate the
regulated volume expansion coefficients as given in Equation~\nn{coeffs}, into 
explicit, geometric, boundary integrals. Firstly, computing the coefficient of the leading divergence
requires only the integrated delta function identity~\nn{deltaintegral}, which leads to the 
holographic formula 
\begin{equation}\label{leading}
v_{d-1}=\frac{1}{d-1}\, \int_\Sigma \left.\left( \frac{1}{\sqrt{\bm{\mathcal S}}\, \bm \tau^{d-1}}\right)\right|_\Sigma\, .
\end{equation}
For the remaining divergences, 
we use Proposition~\ref{delta_deriv} to rewrite the differentiated delta function densities $\bm \delta^{(j)}$ as powers of the Laplace--Robin operator acting on the undifferentiated delta density~$\bm \delta$ and then integrate these by parts  onto  the power of the  regulator $\bm \tau$ using Theorem~\ref{parts}, and finally perform the delta integration according to Equation~\nn{deltaintegral}. 
At this point we consider the case that~$\Sigma$ is closed. Then, the compactly supported test function~$\chi$ 
introduced in Section~\ref{distributions} ensures that the surface terms 
 generated by the total divergence term $\divergence {\bm j}^a$ of the integration by parts Theorem~\ref{parts}, do not contribute. 
We record the result of this computation in the following proposition:

\begin{proposition}\label{divs}
Let $\Sigma$ be a closed hypersurface. Then the divergences $v_k$  in Equation~\nn{coeffs} are given by
\begin{equation}
\label{subleadingdivas}
v_{k{\sss\in\{d-1,\ldots,1\}}}=\frac{ (k-1)!}{(d-2)!\,(d-k-1)!\, k }\, \int_\Sigma\frac{1
}{\sqrt{\bm{\mathcal S}}}\, \big(\!-\D\I2^{-1}\!\big)^{d-k-1}\, \frac{1}{\bm \tau^k}
\, .\end{equation}
\end{proposition}

\begin{remark}
In a setting where one is given a distinguished defining density~$\bm{  \sigma}$ smoothly determined  to all orders (for example this not the case for the singular Yamabe problem dealt with in Theorem~\ref{BIGTHE}),
working in a choice of scale, it is   possible to determine the coefficients of finite terms $v_{k\leq-1}$ 
generated by the error term $\varepsilon {\mathcal R}(\varepsilon)$ in Equation~\nn{vs},
in terms of boundary integrals by using the relation $${\mathcal S}\, \delta^{(j)}(\sigma)=(\nabla_n -2j\rho)\, \delta^{(j-1)}(\sigma)$$ to successively remove derivatives from the delta function in Equation~\nn{coeffs}.
\end{remark}

 \begin{remark}
 When the hypersurface~$\Sigma$ has boundary there are surface terms which can be computed using the  
 result quoted in the theorem for the current ${\bm j}^a$. We reserve that computation for a future work.
\end{remark}

\subsubsection{The anomaly}

We now compute the anomaly.
For that, according to Equation~\nn{preQ}, we need to compute $d-1$ derivatives of the delta function $\delta( \sigma)$.
However, Proposition~\ref{delta_deriv}
is no longer of immediate assistance, since this is the critical case
where $$(\I2^{-1} \D)^{d-1}\bm \delta = 0\, .$$ 
The main idea to resolve this problem is to strategically introduce a logarithm of a true scale.
Indeed even though Equation~\nn{preQ} does
not involve the regulating scale $\bm \tau$,  
by reintroducing some true scale~$\bm \tau$ (which need not coincide with the regulating scale, but for efficiency we lose no generality by recycling this quantity, as the final result for ${\mathcal A}$ is independent of any such choice) we can write a holographic formula for the anomaly.
The following lemma is key:

\begin{lemma}\label{lastderiv}
Let~$\Sigma$ be a closed hypersurface and
$\bm \tau$ be a weight one density, and
suppose the $\bm{\mathcal S}$-curvature is nowhere vanishing, then
\begin{equation}\label{logform}
\int_{\, \widetilde{\! D}}\bm{\delta}^{(d-1)}=
-\int_{\,\widetilde{\! D}}
\bm{\delta}^{(d-2)}
\big(\I2^{-1}\circ \D  -
(\bm \nabla_a \bm {\mathcal S}^{-1}) \, {\bm g}^{ab}
  \triangledown^{\scalebox{.7}{$\bm\sigma$}}_b
 \big)  \log \bm \tau
\, .
\end{equation}
\end{lemma}

\begin{proof}
First note that
 $\log\bm \tau$ is a weight one log density 
 as described in Section~\ref{confdenintro}.
Let us 
work in the scale ${\bm \tau}=[\hh g\, ; 1]$.
From Equations~\nn{Dlog} and~\nn{trianglelog},
and using 
 Equation~\nn{xdj}
 we have 
\begin{equation*}
\begin{split}
\bm{\delta}^{(d-2)}
\big(\I2^{-1}\circ \D & -
(\bm \nabla_a {\mathcal S}^{-1}) \, {\bm g}^{ab}
  \triangledown^{\scalebox{.7}{$\bm\sigma$}}_b
 \big)  \log \bm \tau \\[1mm]
&=\big[\hh g\, ;\,
\delta^{(d-2)}\,  {\mathcal S}^{-1}\big(
 (d-2)\rho  -\sigma \J\,  \big)
 +
 \delta^{(d-2)} \nabla_n {\mathcal S}^{-1}
 \big]
 \\[1mm]
 &=
 \big[\hh g\, ;\, (d-2)\, {\mathcal S}^{-1}(\rho\, \delta^{(d-2)}+\J\, \delta^{(d-3)})
 +
 \delta^{(d-2)} \nabla_n {\mathcal S}^{-1}\big]
\\[1mm]
&=
\Big[\hh g\, ;\, \frac{d-2}{d}\, {\mathcal S}^{-1}\big(-(\nabla.n)\, \delta^{(d-2)}+{2\J}\, (d-1)\delta^{(d-3)}\big)
+
 \delta^{(d-2)} \nabla_n {\mathcal S}^{-1}\Big]\, .
 \end{split}
\end{equation*}
We now concentrate on the divergence of the normal vector term:
\begin{equation*}
\begin{split}
\int_{\,\widetilde{\! D}}\ext \! V^g 
{\mathcal S}^{-1}\, & (\nabla.n)\, \delta^{(d-2)} \ \ =\ \ 
-\int_{\,\widetilde{\! D}}\ext \! V^g 
\big( {\mathcal S}^{-1} n^2\, \delta^{(d-1)} +\delta^{(d-2)}
 \nabla_n {\mathcal S}^{-1} \big)
\\[2mm] 
& =  
-\int_{\,\widetilde{\! D}}\ext \! V^g 
\,\big( (1-2\rho  \sigma  {\mathcal S}^{-1})\, \delta^{(d-1)} 
+\delta^{(d-2)} 
 \nabla_n {\mathcal S}^{-1}\big)
\\[2mm]
&=
-\int_{\,\widetilde{\! D}}\ext \! V^g \,\delta^{(d-1)}
 -\int_{\,\widetilde{\! D}}\ext \! V^g 
 \big(2(d-1)
{\mathcal S}^{-1}
\rho +
 (\nabla_n {\mathcal S}^{-1})
\big)  
\delta^{(d-2)} \\[2mm]
&=
\frac{d}{d-2}\, \int_{\,\widetilde{\! D}}\ext \! V^g \,\big(\delta^{(d-1)}
 +
 (\nabla_n {\mathcal S}^{-1} ) 
\delta^{(d-2)} 
\big)
+\, 2(d-1)\int_{\,\widetilde{\! D}}\ext \! V^g 
{\mathcal S}^{-1}J \,  \delta^{(d-3)}\, .
\end{split}
\end{equation*}
Because we are in the case where~$\Sigma$ is closed, the integrands have no support along~$\partial \,\widetilde{\! D}$. Hence, 
in the first line of the above computation, 
there no  surface terms
generated by an integration by parts.
Combining the above two displays gives the quoted result.
\end{proof}

Computing the anomaly is now simple: 
Proposition~\ref{delta_deriv} can now be used to handle the differentiated delta-density 
$
\bm{\delta}^{(d-2)}$ appearing on the right hand side of~\nn{logform}, 
and thereafter, following the same method employed for the computation of the divergences,   one applies the    integration by parts result of Theorem~\ref{parts}. We record the result in the following theorem:

\begin{theorem}\label{anomalytheorem}
Suppose the $\bm{\mathcal S}$-curvature is nowhere vanishing, and $\Sigma$ is closed, then the anomaly is given by
\begin{equation}\label{Q}
{\mathcal A}=\frac{1}{(d-1)!(d-2)!}\, \int_{\Sigma} 
\, \bm Q^{\scalebox{.7}{$\bm\sigma$}}\, ,
\end{equation}
where
\begin{equation}\label{Qs}
\bm Q^{\scalebox{.7}{$\bm\sigma$}}:=-\frac1{\sqrt{\bm{\mathcal S}}}\, (-\D\I2^{-1})^{d-2}\circ 
\big(\I2^{-1}\circ \D  -
(\bm \nabla_a {\mathcal S}^{-1}) \, {\bm g}^{ab}
  \triangledown^{\scalebox{.7}{$\bm\sigma$}}_b
 \big) 
\,  \log \bm \tau\, \Big|_\Sigma\, .
\end{equation}

\end{theorem}

\begin{remark}
The non-vanishing requirement on the $\bm{\mathcal S}$-curvature 
results in no essential loss of generality since
there must exist a neighborhood of $\Sigma$ where $\bm{\mathcal S}\neq 0$ by virtue of the definition of
a defining density in Section~\ref{DD}. 
\end{remark}

The quantity ${\bm Q}^{\scalebox{.7}{$\bm\sigma$}}$ is a weight $1-d$ density along~$\Sigma$. Moreover it matches
the holographic formula for the Branson $Q$-curvature in the special case of Poincar\'e--Einstein
 structures given in~\cite[Theorem 4.7]{GW} since in that case $\bm{\mathcal S}=1$. 
 
 The integral of $Q$-curvature is a conformal invariant.
 The analogous result
 for  the integral of~$\bm Q^{\scalebox{.7}{$\bm\sigma$}}$ holds here:
 According to Equation~\nn{preQ}, the anomaly~${\mathcal A}$ does not depend on the choice of regulator~$\bm \tau$, so nor does $\int_\Sigma\bm Q^{\scalebox{.7}{$\bm\sigma$}}$ by virtue of the above theorem; but changing the choice of true scale~$\bm \tau$ amounts to changing the choice of metric $g\in \bm c$.
 It is also interesting to construct a direct version of  this argument.
For that we study the behavior of $\bm Q^{\scalebox{.7}{$\bm\sigma$}}$ upon replacing the scale~$\bm \tau$ by $\exp(\bm \varphi)\, \bm \tau$ where $\bm \varphi$ is any smooth weight~$0$ density.
Since $\log \bm \tau$ then becomes $\bm \varphi+ \log \bm \tau$, the corresponding change in $\bm Q^{\scalebox{.7}{$\bm\sigma$}}$ is given by
$$
\bm Q^{\scalebox{.7}{$\bm\sigma$}}\longmapsto 
\bm Q^{\scalebox{.7}{$\bm\sigma$}}+
\widetilde{\rm P}^{\scalebox{.7}{$\bm\sigma$}} \bm \varphi
$$
where the operator~$\widetilde{\rm P}^{\scalebox{.7}{$\bm\sigma$}}$ is given by
$$
\widetilde{\rm P}^{\scalebox{.7}{$\bm\sigma$}}:=
\frac1{\sqrt{\bm{\mathcal S}}}\, (-\D\I2^{-1})^{d-2}\circ 
\big(\I2^{-1}\circ \D  -
(\bm \nabla_a \bm {\mathcal S}^{-1}) \, {\bm g}^{ab}
  \triangledown^{\scalebox{.7}{$\bm\sigma$}}_b
 \big) \, .
$$
For Poincar\'e--Einstein
 structures the above reproduces the holographic formula for the GJMS conformal Laplacian powers presented in~\cite{GW}. For general scales $\bm \sigma$, it amounts to a version of the tangential operator appearing in Equation~\nn{tangID} modified precisely so that $\int_\Sigma \widetilde{\rm P}^{\scalebox{.7}{$\bm\sigma$}} {\bm \varphi}=0$, which implies that that $\int_\Sigma \bm Q^{\scalebox{.7}{$\bm\sigma$}}$ is conformally invariant. To prove this, we use that 
 $$\int_\Sigma \widetilde{\rm P}^{\scalebox{.7}{$\bm\sigma$}} {\bm \varphi}=\int_{\widetilde D} \bm \delta \, \sqrt{\bm{\mathcal S}}\ \widetilde{\rm P}^{\scalebox{.7}{$\bm\sigma$}} {\bm \varphi}
 \propto
 \int_{\widetilde D} \bm \delta^{(d-2)}
 \big(\I2^{-1}\circ \D  -
 (\bm \nabla_a \bm {\mathcal S}^{-1}) \, {\bm g}^{ab}
  \triangledown^{\scalebox{.7}{$\bm\sigma$}}_b
 \big) \bm \varphi\, .
 $$
 The last expression above follows from Proposition~\ref{delta_deriv}. From Theorem~\ref{parts} and Proposition~\ref{Lfdelta}
 we see that the first term on the right hand side above equals $\int_{\widetilde D} \bm \delta^{(d-2)}\bm \varphi
 \D\bm {\mathcal S}^{-1}$. Therefore we must compute the final term of the above display:
 \begin{equation*}
 \begin{split}
 \int_{\widetilde D} \bm \delta^{(d-2)}
  (\bm \nabla_a \bm {\mathcal S}^{-1}) \,& {\bm g}^{ab}
  \triangledown^{\scalebox{.7}{$\bm\sigma$}}_b
 \bm \varphi
 =-\int_{\widetilde D}
dV^g \, \varphi
\nabla_a\big(\sigma 
(\nabla^a {\mathcal S}^{-1})
\delta^{(d-2)}\big)\\[1mm]
&=\int_{\widetilde D} dV^g\,  \delta^{(d-2)}
\varphi 
\big(
(d-2)\nabla_n{\mathcal S}^{-1}
-\sigma \Delta  {\mathcal S}^{-1}
\big)=\int_{\widetilde D} \bm \delta^{(d-2)} \bm \varphi \D\bm {\mathcal S}^{-1}\, .
\end{split}
\end{equation*}
For the first equality above we made a choice of scale and used that we are in the case that $\Sigma$ is closed to integrate the operator $\triangledown^{\scalebox{.7}{$\bm\sigma$}}_b$ by parts without incurring surface terms. The second equality relied on the identity~\nn{xdj} and the final result follows from Equation~\nn{Ldef}. We have therefore proved the following result twice:
\begin{proposition}\label{intQ}
Let $\Sigma$ be a closed hypersurface with defining function~$\bm \sigma$. Then~$\int_\Sigma \bm Q^{\scalebox{.7}{$\bm\sigma$}}$
is a conformal hypersurface invariant depending only on the data of the conformal embedding and the defining density $\bm \sigma$. 
\end{proposition}

\subsection{Asymptotically hyperbolic spaces}\label{AHs}

To illustrate our method's efficacy 
we compute, in terms of standard Riemannian quantities, the anomaly for an almost hyperbolic~3-manifold: Given any conformally compact manifold with boundary~$\Sigma$, there is a conformally related singular metric  $g^o$ with the property that the 
scalar curvature $\Sc^{g^o}$
is non-singular and approaches the strictly negative constant 
$-d(d-1)$ along~$\Sigma$.
In this case $(M,g^o)$ is said to be {\it asymptotically hyperbolic} (AH) and 
$$g^o=\frac{g}{\sigma^2}\, ,
$$
where $g$ is a smooth metric on the manifold with boundary $\, \overline{\! M}$
and $\sigma$ is a defining function for $\Sigma=\partial \overline{\! M}$ such that 
\begin{equation}\label{one1}
|\!\ext \! \sigma|^g\big|_\Sigma=1\, .
\end{equation}
Observe that the $\bm{\mathcal S}$-curvature  of $\bm \sigma=[\hh g\, ;\, \sigma]$ then obeys
\begin{equation}\label{AH}
{\bm {\mathcal S}}|_\Sigma = 1\, ,
\end{equation}
which may also be taken as the definition of asymptotic hyperbolicity.

For a given fixed AH singular metric $g^o$, it is possible to 
find conformal representatives 
for the defining density $\bm \sigma=[\hh g\, ;\, \sigma]$ such that the defining function obeys the unit length condition~\nn{one1} not only along~$\Sigma$, but also in some collar neighborhood thereof. 
Using this defining function  as a coordinate $x$, 
 there exist further coordinates such that the singular metric 
  takes 
 the Graham--Lee normal form~\cite{GL}
$$
g^o=\frac{dx^2+h(x)}{x^2}\, .
$$
Clearly the $\bm{\mathcal S}$-curvature is left unchanged.

It is also possible
via a normal coordinate construction 
to instead fix $g$ and  find a new~AH singular metric with  
defining function obeying~\nn{one1} in a neighborhood of $\Sigma$ (see for example~\cite{Wald} or \cite[Proposition 2.5]{GW}
for an explicit asymptotic construction). The $\bm{\mathcal S}$-curvature then still obeys the AH condition~\nn{AH} but is changed  away from~$\Sigma$.

In the following example both situations are covered:  We assume that
the defining density
$\bm \sigma=[\hh g\, ;\, \sigma]$
  obeys the  AH condition~\nn{AH}, and  $n=\nabla\sigma$ obeys $|n|^g=1$ in some neighborhood of $\Sigma$ in which we now work. 
In $d=3$ dimensions  
the $\bm{\mathcal S}$-curvature of $\bm \sigma$ is then given by
$$\bm {\mathcal S}=[\hh g\, ;\, 1 + 2\rho\sigma]\, ,$$
where $\rho=-\tfrac13(\Delta \sigma+J\sigma)$.
Now we consider the weight one density $\bm \tau=[\hh g\, ;\, 1]$ determined by $g$ and a weight zero density $\bm f=[\hh g\, ;\, e^\varphi]$
so that $\bm f \bm \tau$ may be viewed as an arbitrary true scale.
We want to compute~$\bm Q^{\scalebox{.7}{$\bm\sigma$}}$ as given in Equation~\nn{Qs}:
\begin{equation*}
\begin{split}
-\frac1{\sqrt{\bm{\mathcal S}}}\, (-&\D\I2^{-1})^{d-2}\circ 
\big(\I2^{-1}\circ \D  -
(\bm \nabla_a \bm {\mathcal S}^{-1}) \, {\bm g}^{ab}
  \triangledown^{\scalebox{.7}{$\bm\sigma$}}_b
 \big) 
\,  \log (\bm f \bm \tau)\, \Big|_\Sigma\\
&=\Big[\hh g\, ;\, 
-(\nabla_n-\rho)\left( {\mathcal S}^{-2^{\phantom{o'}}}\!\!\!
\big(\nabla_n \varphi+\rho
-\sigma \Delta \varphi
-\sigma J
+(\nabla^a \log{\mathcal S})
(\sigma\nabla_a \varphi -
n_a)
\big)\right)\Big]\Big|_\Sigma\\[1mm]
&
=\Big[\hh g\, ;\, 
-(\nabla_n-5\rho)
\big(
\nabla_n \varphi+\rho
-\nabla_n \log{\mathcal S}
\big)
+\Delta\varphi-(\nabla^a \log{\mathcal S})\nabla_a\varphi
+J
\Big]\Big|_\Sigma
\\[1mm]
&=
\Big[\hh g\, ;\, 
\bar\Delta \varphi
-\nabla_n\rho
+5\rho^2
+(\nabla_n^2 -5\rho \nabla_n)\log {\mathcal S}
+J
\Big]\Big|_\Sigma\, .
\end{split}
\end{equation*}
In the above we used that $|n|=1$ implies that  $(\Delta \varphi
+3\rho \nabla_n \varphi - \nabla_n^2 \varphi)
|_\Sigma=\bar \Delta \varphi$.
Thus we see that $\bm Q^{\scalebox{.7}{$\bm\sigma$}}$ depends on $\varphi$ only through the hypersurface total divergence $\bar \Delta \varphi$ so that, in concordance with Proposition~\ref{intQ}, its integral along $\Sigma$ is independent of the choice of regulator $\bm f\bm \tau$.   
We still wish to express the remaining terms 
as curvatures:
\begin{equation*}
\begin{split}
-\nabla_n\rho
+5\rho^2
&+(\nabla_n^2 -5\rho \nabla_n)\log {\mathcal S}
+J\stackrel\Sigma=
3\nabla_n\rho
-9\rho^2
+J
\stackrel\Sigma=
K+P_{\hat n\hat n}-H^2+J\, .
\end{split}
\end{equation*}
Here we used that $|n|=1$ implies that $\nabla_a n_b|_\Sigma=\II_{ab}$, $\rho|_\Sigma=-\tfrac 23 H$ and  $
\nabla_n\rho|_\Sigma=H^2+
\tfrac13 (K+P_{\hat n\hat n})
$.
Using the hypersurface identity $J=\bar J+P_{\hat n\hat n}-H^2+\tfrac12 K$ we have the general result for the anomaly in almost hyperbolic 3-space
$$
{\mathcal A}=\frac12
\int_\Sigma \bm Q^{\scalebox{.7}{$\bm\sigma$}}
=\int_\Sigma dA^{\bar g} \Big(\bar J-\frac{K}2\Big)
+
\int_\Sigma dA^{\bar g} \big(P_{\hat n\hat n}+ K-H^2\big)\, .
$$
The rigidity density $K$ is a conformal hypersurface invariant
and the integral over $\bar J$ is proportional to the Euler characteristic of $\Sigma$, so the first  term is an invariant of the  conformal embedding. In fact, 
the integrand of the second term equals $\tfrac34\big(\nabla_n^2 {\mathcal S}-
3(\nabla_n{\mathcal S})^2\big)\big|_\Sigma$.
Thus imposing a  condition $\bm {\mathcal S}=1+{\mathcal O}(\bm \sigma^3)$, the anomaly would then be an invariant of the conformal embedding $\Sigma\hookrightarrow (M,\bm c)$. This further motivates the singular Yamabe problem studied in the next section.

%

\section{The singular Yamabe problem}\label{singYam}

The volume $\Vol(\, \widehat{\!D};{\bm \sigma})$ defined in Equation~\nn{VolDhat}
depends on the choice of defining density~${\bm \sigma}$, or equivalently the bulk metric $g^o$. When given only the conformal embedding~$\Sigma\hookrightarrow (M,\bm c)$,
there is a  canonical choice of defining density (determined up to the order required to compute a renormalized version of the volume integral). The divergences simplify considerably in that setting, and the anomaly is an invariant of  the conformal embedding~\cite{Grahamnew}.
Indeed, the singular Yamabe problem underlies a general program for  the study of conformal hypersurface invariants~\cite{CRMouncementCRM,GW}.

On a compact manifold, every metric is conformal to a metric of constant scalar curvature. On closed manifolds, the problem of finding a conformal rescaling $\Omega$ such that $\Sc^{\Omega^2\SSmidge g}$ is constant is called the {\it Yamabe problem}. We term the   analogous problem for conformally compact manifolds 
 the  {\it singular Yamabe problem} ({\it cf.}~\cite{MazzeoC}). This is formulated simply in terms of the ${\bm {\mathcal S}}$-curvature:

Firstly consider an arbitrarily chosen defining density ${\bm \sigma}_0=[\hh g\, ;\, \sigma_0]$. Then since $\ext \!\sigma_0$ is non-vanishing along~$\Sigma$, its ${\bm {\mathcal S}}$-curvature is positive in a neighborhood of~$\Sigma$. Hence, at least in this neighborhood of~$\Sigma$, the new defining density ${\bm \sigma}={\bm \sigma}_0/\sqrt{\bm {\mathcal S}(\bm \sigma_0)}$ is well-defined and its ${\bm {\mathcal S}}$-curvature obeys the AH condition
$$
{\bm {\mathcal S}}|_\Sigma = 1\, .
$$
In the following discussion,   let us assume that the chosen defining density ${\bm \sigma}$ obeys the unit  property in the above display. Now suppose it were also possible to choose ${\bm \sigma}$ such that the unit property held throughout $M$, namely
$$
1=: [\hh g\, ;\, 1]={\bm {\mathcal S}}= [\hh g\, ;\, {\mathcal S}]\, .
$$
Then in the interior $\,\widehat{\! M}$,  evaluating ${\mathcal S}$ in the scale $g^o$ (so that $\sigma=1$) we would have
$
1=-\frac{2\smallJ^{\, g^{\scriptscriptstyle o}}}{d}
$,
corresponding to an interior metric with constant negative scalar curvature $$\Sc^{g^o}=-d(d-1)\, .$$
Hence the singular Yamabe problem amounts to finding smooth defining densities such that 
$${\bm {\mathcal S}}=1\, .$$

In general   smooth solutions to the singular Yamabe problem do not exist~\cite{ACF}. However, approximate solutions to sufficiently high orders to define a renormalized volume do exist, as encapsulated by the following theorem (based in part on~\cite{ACF}):

\begin{theorem}[\cite{GW15}]\label{BIGTHE}
Given a defining density~${\bm \sigma}_0$, there exists an improved defining density \begin{equation}\label{expansion}\bm{\bar\sigma}={\bm \sigma}\big(1+{\bm \alpha_1}\bm \sigma + \cdots +{\bm \alpha_{\db}}{\bm \sigma}^{\db}\big)\, ,\end{equation}
where $\bm\sigma=\bm \sigma_0/{\sqrt{\bm {\mathcal S}(\bm \sigma_0)}}$ in a neighborhood of $\Sigma$, and
${\bm \alpha}_k$ are smooth densities, such that the ${\bm {\mathcal S}}$-curvature of $\bm{\bar\sigma}$ obeys
\begin{equation}\label{one}
\bm {\mathcal S}=1+{\bm \sigma}^d {\bm {\mathcal B}}\, .
\end{equation}
Moreover, the weight~$w=-d$ density $\bm {\mathcal B}({\bm \sigma}_0)$ is a preinvariant
for a  conformal hypersurface invariant 
$${\bm B}:=\bm{\mathcal B}|_\Sigma\, ,$$
which depends only on the data of the conformal embedding~$\Sigma\hookrightarrow (M,\bm c)$.
\end{theorem} 

The density $\bm{\bar\sigma}$ of the theorem is unique up to terms $\bm \sigma^{d+1} \bm \alpha$, where $\bm \alpha$ is a smooth weight~$-d$ density, and~$\bm{\bar\sigma}$ is termed a {\it conformal unit defining density}. Since the  density~${\bm B}$ obstructs smooth solutions to the singular Yamabe problem, it is called the {\it obstruction density}. For surfaces embedded in conformally Euclidean 3-space,  in a Euclidean scale~$3B$ equals the {\it Willmore invariant}
\begin{equation}\label{EWillmore}
\bar\Delta H+2H(H^2-\kappa)\, ,
\end{equation}
where $\kappa=\, \overline{\!\Sc}/2$ is the Gau\ss\  curvature. It follows that the above quantity, which appears as one side of the Willmore equation, is 
invariant under rigid conformal motions, a fact which is well known.

\subsection{Divergences}

It is not difficult to generate general formul\ae\ for the divergences in the regulated volume~\nn{vregd} for singular metrics solving the  singular Yamabe problem. 
We  focus on the case where $\Sigma$ is closed throughout this section.
Computations are simplified by working in the  scale~$\bm \tau$. For the leading divergence of Equation~\nn{leading}, 
this yields the hypersurface integral\begin{equation}\label{leadingSINGYAM}
v_{d-1}=\frac{1}{d-1}\, \int_\Sigma \ext \! A^{\bar g}\, .
\end{equation}
In a Poincar\'e--Einstein setting this behavior of the leading divergence is well known, see for example~\cite{Gra00}.

By virtue of Equation~\nn{one},  the subleading divergences~\nn{subleadingdivas} become
\begin{equation}\label{divergences}
v_{k{\sss \in\{d-1,\ldots,1\}}}=
\frac{(k-1)!}{(d-2)!\,(d-k-1)!\, k }
\, \int_\Sigma \, 
(-\D)^{d-k-1}\, \frac{1}{\bm \tau^k}\, ,
\end{equation}
and it is not difficult to develop explicit formul\ae\ for the first few values of $k$:
Using  Equation~\nn{Ldef}   we compute
\begin{eqnarray*}
\D \bm \tau^{2-d}&\stackrel\Sigma=& [\hh g\, ;\, (d-2)^2 \rho\, 
]\, ,
\end{eqnarray*}
where again $\bm \tau=[\hh g\, ;\, 1]$.
Thus, using Lemma~\ref{Hlemma}, it follows that in this scale the coefficient of the next-to-leading order (nlo) divergence is
\begin{equation}\label{next-to-leading}
v_{d-2}=\int_\Sigma \ext \! A^{\bar g} \, H\, .
\end{equation}

To compute the next-to-next-to-leading order (nnlo) divergence
we must calculate $
\D^2 \bm \tau^{3-d}|_\Sigma$. The geometric data required for this computation 
 is given in~\cite[Lem\-ma~7.9]{GW}. In the scale  $\bm \tau=[\hh g\, ; \, 1]$, using Equation~\nn{K2J}, we then find
\begin{equation}\label{L2taui}
\D^2 \bm \tau^{3-d}\stackrel\Sigma=-(d-2)(d-3)\Big[\hh g\, ;\, 
\J  \, +\, (d-4)\Big(\Rho^{ab}\hat n_a\hat n_b-(d-2)H^2
 \, +\,\frac{K}{d-2}\, \Big)
\Big]\, .
\end{equation}
Thus, in this scale,  the coefficient of the nnlo  divergence is
\begin{equation}\label{nnlo}
v_{d-3}=-\frac1{2(d-3)}\, \int_\Sigma \ext \! A^{\bar g}\, 
\Big\{\J  \, +\, (d-4)\Big(\Rho^{ab}\hat n_a\hat n_b-(d-2)H^2
 \, +\,\frac{K}{d-2}\, \Big)
\Big\}
\, .
\end{equation}

The computation of the next-to-next-to-next-to-leading (nnnlo) divergence is somewhat more involved and has been relegated to Appendix~\ref{nnnlo}.

To summarise, 
given  the data of a compact hypersurface embedded in a Riemannian manifold~$(M,g)$
and the corresponding choice of true scale $\bm \tau:=[\hh g\, ;\, 1]$,  the regulated volume for a conformal unit defining density
is given by
%
%
\begin{equation}\label{firsttwo}
\Vol_\varepsilon(D,\Sigma)\ =\ \frac1{d-1}\,  \frac{\int_\Sigma \ext \! A^{\bar g}}{
\varepsilon^{d-1}}\ +\ 
\frac{\int_\Sigma \ext \! A^{\bar g} H}{\varepsilon^{d-2}}\ +\ \cdots\   +\ 
{\mathcal A}\,  \log\varepsilon
\ +\  \Vol_{\rm ren}\ +\  {\mathcal O}(\varepsilon)\, .
\end{equation}
Formul\ae\  for the nnlo $1/\varepsilon^{d-3}$ and nnnlo $1/\varepsilon^{d-4}$ divergences can be found in Equation~\nn{nnlo}
and Appendix~\ref{nnnlo}.

\subsection{The anomaly}

We can also generate explicit results for the anomaly in the singular Yamabe setting. These 
are of particular interest, since they generate integrated conformal invariants depending only on the conformal embedding.  
 
 To begin with note that Theorem~\nn{anomalytheorem}
 simplifies considerably for defining densities satisfying Equation~\nn{one}. In this case $\bm Q$ is independent of $\Sigma$ and given by $(-\D)^{d-1} \log \bm \tau \, \big|_\Sigma$.
 Thus, for closed $\Sigma$, the anomaly is 
 given by
 $$
{\mathcal A}=\frac{1}{(d-1)!(d-2)!}\, \int_\Sigma (-\D)^{d-1} \log \bm \tau \, \Big|_\Sigma\, .
$$
We now develop the above formula for embedded surfaces and spaces.

 \subsubsection{Surfaces embedded in 3-manifolds}

To compute the $\log$ term in the $\varepsilon$ expansion of the regulated volume  when the host space is three dimensional we need to compute the square of 
the Laplace--Robin operator acting on a log-density.  
An explicit formula for the square of the Laplace--Robin operator of a conformal unit density acting on general densities (and tractors) is known (see for example~\cite[Lemma 7.9]{GW15})
and is given by 
$$
\D^2  \log\bm \tau\stackrel\Sigma=\Big[\hh g\, ;\, 
\, \bar{\!\!\J}-\frac K2\, \Big]\, .
$$
Orchestrating the above with our results for the leading and subleading divergences
in Equations~\nn{leading} and~\nn{next-to-leading},
the regulated volume in the scale $\bm \tau$ is given by
\begin{equation*}
\Vol_\varepsilon(D,\Sigma)\ =\  \frac{\int_\Sigma \ext \! A^{\bar g}}{2\varepsilon^2}\ +\ 
\frac{\int_\Sigma \ext \! A^{\bar g} H}{\varepsilon}\ +\ \frac{\log\varepsilon}2\, \int_\Sigma \ext \! A^{\bar g}\Big(\ \bar{\!\!\J} -\frac K2 \, \Big)  \ +\  \Vol_{\rm ren}\ +\  {\mathcal O}(\varepsilon)\, .
\end{equation*}
Thus, remembering that the Gau\ss\  curvature equals $\bar J$,  we see that the anomaly for closed hypersurfaces and singular metrics defined by a conformal unit defining density
is
\begin{equation}\label{Willd=2}
{\mathcal A}=\pi \chi -\frac 14 
\, \int_\Sigma \bm K\, , 
\end{equation}
where the Euler characteristic 
$\chi$ of $\Sigma$ is clearly conformally invariant and the rigidity density is a local conformal hypersurface invariant. 
Hence ${\mathcal A}$ depends only on the conformal embedding.
Of course, the integral of intrinsic scalar curvature does not contribute to the functional gradient of ${\mathcal A}$ so that
for Euclidean ambient spaces, 
  Equation~\nn{K2J}
  shows that the only remaining variational term is
  the integral of mean curvature-squared, or in other words the classical Willmore energy functional.

\subsubsection{Spaces embedded in 4-manifolds}

Here we need the square of the Laplace--Robin operator acting on $1/{{\bm \tau}}$ and its cube acting on a log-density. The former is given in Equation~\nn{L2taui} and for $d=4$ yields
$$
\D^2\frac1{\bm \tau}\, \stackrel\Sigma=\, [\hh g\, ;\,  -2\J\, 
]\, .
$$
The cubic computation is more involved although significantly simplified by calculating in the $\bm \tau$ scale. First we use the definition of the Laplace--Robin operator in Equations~\nn{Ldef},~\nn{Dlog} and find along the hypersurface
$$
\D^3\,\log\bm \tau\stackrel\Sigma=\big[\hh g\, ;\, \big(2\nabla_n\circ\bar\sigma(\Delta-\J\, )\big)(2\rho -\bar\sigma \J\, )\big]\stackrel\Sigma=[\hh g\, ;\, 2(\Delta-\J\, )(2\rho -\bar\sigma \J\, )]\, .
$$
The second equality above used that for a conformal unit defining density $\nabla_n\bar\sigma\stackrel\Sigma=1$. Furthermore,
along $\Sigma$ we also have the operator identities (see~\cite{GW15})
\begin{equation}\label{box}\Delta\circ \bar\sigma \stackrel\Sigma=
2\nabla_n+d\, H\, \mbox{ and }\, 
\Delta\stackrel\Sigma=
\bar\Delta+\nabla_n^2+(d-2)H\nabla_n\, .
\end{equation}
Hence
$$
\D^3\,\log\bm \tau\stackrel\Sigma=[\hh g\, ;\,
-4\, \bar\Delta H
+4\, \nabla_n^2\rho
+8H\nabla_n\rho
-4(\nabla_n+H)\J
\, ]\, .
$$
The quantities appearing above have been computed in~\cite[Lemmas 6.6 \& 6.8]{GW15},
in particular
\begin{equation}\label{1}
\nabla_n\rho\stackrel\Sigma=\Rho_{ab}n^an^b\, +\, \frac{K}{d-2}\, ,
\end{equation}
and
\begin{eqnarray}
\qquad \nabla_n^2\rho&\stackrel\Sigma=&
(\nabla_n + H)\J 
-\frac{1}{(d-2)(d-3)}\, \Big(
\nablab^a\nablab^b\, \IIo_{ab}+(d-2)(d-4)\, \IIo^{ab}\bar\Rho_{ab}
\Big)
\label{2}
\\[2mm]
&&-\ \frac{d-2}{d-3}\ \,  \IIo^{ab}\, {\mathcal F}_{ab}
-\ 
\nablab^a(\Rho_{ab}n^b)^\top
-H\big((d-2)\Rho_{ab}n^an^b
+K\big)\, .
\nonumber
\end{eqnarray}
Also, we have the hypersurface identity
\begin{equation}\label{3}
\bar\Delta H=\frac{1}{d-2}\,\nablab^a\nablab^b\, \IIo_{ab}-\nablab^a(\Rho_{ab}n^b)^\top\, .
\end{equation}
Orchestrating the above gives
$$
\D^3\,\log\bm \tau\stackrel\Sigma=[\hh g\, ;\,-4\nablab^a\nablab^b \, \IIo_{ab}-8\, \IIo^{ab}\, {\mathcal F}_{ab}]\, .
$$
Up to the leading divergence (and so non-variational) term, this matches the higher Willmore energy density found for embedded spaces in~\cite{GGHW15}.

A useful check of our result is the linear shift property of the $Q$-curvature discussed in the introduction: For that, notice that under a conformal transformation $g\mapsto \Omega^2g$, we have 
$$
\nablab^a \nablab^b \, \IIo_{ab}\mapsto \Omega^{-3} \big[(\nablab^a \nablab^b \, \IIo_{ab})
+2\big(\IIo^{ab}\nablab_a\nablab_b+(\nablab_a \, \IIo^{ab})\nablab_b
\big)
\log \Omega\big]\, ,
$$
which implies the correct shift transformation:
$$
Q:=4\, \big(\nablab^a\nablab^b \, \IIo_{ab}+2\, \IIo^{ab}\, {\mathcal F}_{ab}\big)\mapsto
\Omega^{-3} \big(Q - \GJMS_3 \log\Omega)\, ,
$$
where the third order extrinsic conformal ``Laplacian power'' $P_3$ acts on weight zero, scalar densities in host dimension $d=4$ according to
$$\GJMS_3:=-8\, \IIo^{ab}\nablab_a\nablab_b-8\, (\nablab_a \, \IIo^{ab})\nablab_b
=-8\nablab_a\circ \IIo^{ab}\circ 
\nablab_b\, ,$$
see~\cite[Proposition 8.5]{GW15}. 
Like the standard, even dimension parity GJMS Laplacian powers of~\cite{GJMS}, the above operator
is formally self adjoint and annihilates constant functions. 
In using the term Laplacian power for odd dimensional hypersurfaces, we  view the trace-free second fundamental form as a metric-like tensor.

Altogether, the regulated volume in the scale~$\bm \tau$ reads
\begin{equation}\label{d=3Will}
\begin{split}
\Vol_\varepsilon(D,\Sigma)\ =\  \frac{\int_\Sigma \ext \! A^{\bar g}}{3\varepsilon^3}\: &+\: 
\frac{\int_\Sigma \ext \! A^{\bar g} H}{\varepsilon^2}\: -\: 
\frac{\int_\Sigma \ext \! A^{\bar g} \J
}{2\varepsilon}\\[2mm] &+\: \frac13 \,
\log\varepsilon \, \int_\Sigma \ext \! A^{\bar g}\Big(
\nablab^a\nablab^b \, \IIo_{ab}+2\, \IIo^{ab}\, {\mathcal F}_{ab}\Big)  \ +\  \Vol_{\rm ren}\ +\  {\mathcal O}(\varepsilon)\, .
\end{split}
\end{equation}

\section{The functional gradient}\label{functgrad}

We now use our boundary calculus to compute the variation of the anomaly. This confirms 
 the result of~\cite{Grahamnew} that the functional gradient of the anomaly~${\mathcal A}$, for singular metrics determined by a conformal unit defining density, is the obstruction density~$\bm{ B}$. More precisely:
\begin{equation}\label{varn}
\frac{\, \updelta{\mathcal A}\, }{\updelta\Sigma}\, = \, \frac{d(d-2)}{2} \, \bm B\, ,
\end{equation}
where $\updelta {\mathcal A}/\updelta\Sigma$ denotes the functional gradient with respect to  variations of the embedding of the hypersurface~$\Sigma$.

In~\cite{GGHW15}, a holographic approach for variations of embeddings was developed and exploited for computations of higher Willmore energy variations. This method is well adapted to the current situation where our starting point is the bulk integral expression in Equation~\nn{regvol} for the regulated volume.
The main idea of the method is as follows:
Given a  functional $E(\Sigma)=\int_\Sigma P$ where $P$ is a hypersurface invariant, we first express $\Sigma$ as the zero locus of a defining function $\sigma_0$ and $P$ as the restriction of a preinvariant  ${\mathcal P}$. As explained in Section~\ref{Dirac}, we can then express $E(\Sigma)$ as a bulk integral 
$$E[\sigma_0]=\int_{\, \widetilde{\! D}} \ext \! V^g \delta(\sigma) \sqrt{{\mathcal S}} \, \mathcal P\, ,$$
where~$\mathcal S$ is the $\bm{\mathcal S}$-curvature of $\bm \sigma_0=[\hh g\, ;\, \sigma_0]$ in a scale $g$. 
Then the embedding can be varied by functionally varying the defining function~$\sigma_0$. 
For that we introduce a smooth one-parameter family of hypersurfaces $\Sigma_t$ such 
$\Sigma_0=\Sigma$ and $\Sigma_t=\Sigma$ outside some compactly supported region. We also define the variational operator
$\updelta(\cdot):=\frac{d\, \cdot\, }{dt}\big|_{t=0}$. The functional gradient  is then defined by 
$$
\updelta E[\sigma_0]=:\int_\Sigma \ext \! A^{\bar g}\,  \hat\updelta \sigma_0\ 
\frac{\updelta{\mathcal A}}{\updelta\Sigma}\ ,
$$
where $\hat\updelta \sigma_0$ is the hypersurface invariant defined by the preinvariant $\updelta \sigma_0/|\nabla \sigma_0|$ (this is the variational analog of the preinvariant formula for the unit normal in Equation~\nn{unorm}).

For conformal hypersurface invariants defined in terms of the jets of a conformal unit defining density, there is one further useful simplification afforded by the holographic variational calculus.
Namely, the uniqueness property of conformal unit defining densities ({\it i.e.},  $\bar\sigma(\sigma_0)$ see~\cite[Theorem 4,5]{GW15}) ensures that the integrand $Q[\bar\sigma(\sigma_0)\big]$ of the anomaly is a preinvariant.  Since  the  functional derivative $\updelta\bar \sigma(\sigma_0)/\updelta\sigma_0$  along $\Sigma$ is given by $\updelta\sigma_0/|\nabla \sigma_0|$ (this follows directly from the functional dependence $\bar \sigma(\sigma_0)$ implied by the expansion in Equation~\nn{expansion}), we have
$$\updelta\bar \sigma\stackrel\Sigma=\frac{\updelta \sigma_0 }{
|\nabla \sigma_0|}\, = \hat \updelta \sigma_0\, .$$ 
%
%
Hence the functional gradient can be computed by functionally varying~$\bar\sigma$. Our strategy, therefore is to consider the one parameter family of regulated volume integrals
$$\Vol_\varepsilon(D,\Sigma_t)=
 \Vol_\varepsilon(D,\bar \sigma_t)=\int_{\, \widetilde{\! D}}\frac{\ext \! V^g}{\bar\sigma_t^d}\, \theta(\bar\sigma_t/\tau-\varepsilon)
$$
corresponding to conformal unit densities $\bm{\bar\sigma}_t$ of hypersurfaces $\Sigma_t$.
Then, since we have already shown that $\Vol_\varepsilon$ is the sum of a Laurent series in~$\varepsilon$ plus $\log \varepsilon$ times the anomaly, we need only compute the $\log \varepsilon$ contribution to  $\updelta=\frac{d\, \cdot\, }{dt}\big|_{t=0}$ of the above expression.

%
%
%
%
%

\subsection{Varying the defining density}

The variation of the regulated volume breaks into two terms
$$
\updelta \Vol_\varepsilon(D,\Sigma)=-d
\int_{D} {\ext \! V^g}\, \frac{\updelta \bar \sigma}{\bar \sigma^{d+1}}
\, \theta(\bar\sigma/\tau-\varepsilon)\ + \ 
\int_{\, \widetilde{\! D}} \frac{\ext \! V^g}{\bar \sigma^d}\, 
\frac{\updelta \bar \sigma}{\tau}\, \delta(\bar \sigma/\tau-\varepsilon)\, . 
$$
By performing the delta function integration,
the second term can be rewritten as $\varepsilon^{-d}$ multiplying a hypersurface integral:
$$
\frac{1}{\varepsilon^d}\, \int_{\Sigma_\varepsilon}\ext \! A^{\bar g_\varepsilon}\, 
\frac{\updelta \bar\sigma}{\tau^{d+1}}
\, .
$$
Since this hypersurface integral depends smoothly on $\varepsilon$ and is well-defined at $\varepsilon=0$, the above display yields some Laurent series in~$\varepsilon$ and does not produce a $\log\varepsilon$ contribution. 
Hence we must focus on the first term in the functional gradient above. For this we will need a pair of lemmas.

\begin{lemma}\label{Ldsigma}
Let $\bm {\bar\sigma}_t$ be a smooth one parameter family of conformal unit defining densities with $\bm {\bar\sigma}_0=\bm {\bar\sigma}$.
Then the variation $\updelta \bm{\bar\sigma}$ obeys
$$
\D \updelta \bm{\bar\sigma} = \frac{d^2}2\, \bm {\bar \sigma}^{d-1} \Big(  \bm{\mathcal B} \, \updelta \bm {\bar \sigma}+\frac1d\, \bm {\bar \sigma}\, \updelta \bm {\mathcal B}\Big)\, .
$$
\end{lemma}

\begin{proof}
The key is to vary the defining relation for a conformal unit defining density
$$
\bm {\mathcal S}_t=1+\bm {\bar \sigma}^d_t \, \bm {\mathcal B}_t\, . 
$$
The variation of the $\bm{\mathcal  S}$-curvature is easily computed
\begin{eqnarray*}
\updelta\bm {\mathcal S}&=&\Big[\hh g\, ;\ 
2 g^{ab} (\nabla_a \bar \sigma) \nabla_b \updelta \bar \sigma
-\frac{2\updelta \bar\sigma}{d}(
\Delta +\J\, )\bar \sigma
-\frac{2\bar\sigma}{d}(\Delta+\J\, )\updelta\bar\sigma\Big]\\[2mm]
&=&
\Big[\hh g\, ;\, 
2\, \Big(\nabla_n+\rho-\frac{\bar \sigma}d(\Delta+\J\, )\Big)\, \updelta \bar\sigma\Big]=\frac{2}{d}\D \updelta \bm{\bar \sigma}\, ,
\end{eqnarray*}
while the variation of the right hand side
 is $d\,\bm{\bar\sigma}^{d-1} \bm{\mathcal B}\, \updelta \bm{\bar\sigma}+
\bm{\bar \sigma}^d \, \updelta\bm{\mathcal B}$.
\end{proof}

Because $\updelta \bm{\bar \sigma}$ is a weight~1 density, it is not difficult to verify (see~\cite[Lemma 3.1]{GW}) that  the algebra~\nn{algebra} implies that
$$
\D\Big( \frac{\updelta\bm{\bar\sigma}}{\bm{\bar\sigma}^{d+1}}\Big)=
 \frac{\D \updelta\bm{\bar\sigma}}{\ \bm{\bar\sigma}^{d+1}\, }\, ,
$$
whence via Lemma~\ref{Ldsigma} we have
\begin{equation}\label{LRvar}
\D\Big( \frac{\ \updelta\bm{\bar\sigma}\,}{\bm{\bar\sigma}^{d+1}}\Big)
 = \frac {d^2}{2\, \bm{\bar\sigma}^2}\,  \big( \bm{\mathcal B} \, \updelta \bm {\bar \sigma}
 +\frac1d\, \bm {\bar \sigma}\, \updelta \bm {\mathcal B} \big)\, .
\end{equation}

The second lemma relates the left hand side of the above display to~$\updelta\bm{\bar\sigma}/\bm{\bar\sigma}^{d+1}$. 
\begin{lemma}
Let $\bm f$ be a weight $-d$ density. Then (for any defining density $\bm \sigma$), 
$$
\bm{\mathcal S}\,  \bm f \, =\,  \frac{\bm\sigma}{d}\ \! \D \bm f+\big[\hh g\, ;\, \nabla_cj^c\big]\, ,
$$
where $j^c=\frac1d \big(\sigma \nabla^c (\sigma f)+(d-1) \sigma(\nabla^c\sigma)f\big)$.
\end{lemma}
\begin{proof}
Firstly, since $(d+2\w)\bm f=-d \bm f$, the algebra~\nn{algebra} implies
$$
\bm \sigma \D \bm f = d\, \bm{\mathcal S} \, \bm f + \D (\bm \sigma  \bm f)\, .
$$
Then Theorem~\ref{parts} applied to densities $1$ and $\bm \sigma \bm f$ yields
$$
\D(\bm \sigma  \bm f)+\big[\hh g\, ;\, \nabla_c \big(\sigma\nabla^c(\sigma f)+(d-1)\,  \sigma n^c f
\big)\big]=0\, .
$$
Here we used the identity $\D 1=0$.
\end{proof}

For the case of a conformal unit defining density and $\bm f=\updelta \bm {\bar \sigma}/\bm {\bar \sigma}^{d+1}$, applying this lemma to Equation~\nn{LRvar} and subsequently using Equation~\nn{one}
gives
\begin{equation}\label{dsos}
  \frac{\ \updelta\bm{\bar\sigma}\,}{\bm{\bar\sigma}^{d+1}} =\frac{d-2}{2}\,  \bm{\mathcal B}\,  \, \frac{\updelta \bm {\bar \sigma}}{\bm {\bar \sigma}}+ \frac12\,  \updelta \bm {\mathcal B}
 +\big[\hh g\,;\, \nabla_c j^c\big]\, ,\end{equation}
where 
$$
j^c=\frac1d\left(\bar \sigma\,  \nabla^c\Big(\frac{\updelta\bar \sigma}{\bar \sigma^d}\Big)+(d-1)\, \frac{ n^c\, \updelta\bar\sigma}{\bar\sigma^d} \right)\, .
$$
The first term on the right hand side of~\nn{dsos} will be responsible for the
$\log \varepsilon$ contribution. Before studying it in detail, we first establish that the other two terms can only produce Laurent series contributions:
The obstruction density and therefore its variation are regular along~$\Sigma$ while the $\bm {\mathcal S}$-curvature is unity there. Hence the second term on the right hand side of  Equation~\nn{dsos} can only produce terms analytic in~$\varepsilon$.
For the total divergence term we employ Green's theorem, 
$\int_{D_\varepsilon} \ext \! V^g \, \nabla_a j^a=\int_{\partial D_\varepsilon} \ext \! A^{\bar g_{_{\scalebox{.8}{$\sss\partial\! D_\varepsilon$}}}}\,  \hat n_a^{\scalebox{.8}{$\sss\partial\! D_\varepsilon$}} j^a$ where $\hat n_a^{\scalebox{.8}{$\sss\partial\! D_\varepsilon$}}$ is the unit outward normal. We thus find a  contribution to the variation proportional to
$$
\int_{\Sigma_\varepsilon} \ext \! A^{\bar g_\varepsilon}\,  \hat n_c^\varepsilon \, \left(\bar \sigma\,  \nabla^c\Big(\frac{\updelta\bar \sigma}{\bar \sigma^d}\Big)+(d-1)\, \frac{ (\nabla^c\bar\sigma)\, \updelta\bar\sigma}{\bar\sigma^d} \right)\, ,
$$
where we have dropped the contribution from the surface term integrated over~$\partial D_\varepsilon\backslash \Sigma_\varepsilon$ as this term is   not responsible for a $\log\varepsilon$ contribution. In the above display the outward unit normal vector to $\Sigma_\varepsilon$ is given by
$$
\hat n^\varepsilon_a=-\left.\frac{\nabla_a(\bar\sigma/\tau)}{|\nabla(\bar\sigma/\tau)|}\right|_{\Sigma_\varepsilon}=\left.-\frac{n_a -\varepsilon\nabla_a \tau}{\sqrt{n^2-2\, \varepsilon \nabla_n\tau +\varepsilon^2 |\nabla \tau|^2}}\, \right|_{\Sigma_\varepsilon}\, ,\qquad n_a:=\nabla_a\bar \sigma, 
$$ 
because $\bar\sigma/\tau-\varepsilon$ is a defining function for $\Sigma_\varepsilon$.
Since $n_a$ is well-defined along $\Sigma$, it follows that $\hat n^\varepsilon_a$ is regular around $\varepsilon=0$. Furthermore, along $\Sigma_\varepsilon$ we have
$$\bar \sigma\,  \nabla^c\Big(\frac{\updelta\bar \sigma}{\bar \sigma^d}\Big)+(d-1)\, \frac{ (\nabla^c\bar\sigma)\, \updelta\bar\sigma}{\bar\sigma^d} =
-
\frac1{\varepsilon^d}\,\frac{n^c\, \updelta\bar\sigma-\varepsilon\, \tau\,  \nabla^c \updelta\bar\sigma}{\tau^d}\, .
$$
Since $\updelta\bar \sigma$ and $n^c$ are regular as $\varepsilon$ approaches zero and 
$\bm \tau$ is a true scale, the above is a Laurent series in $\varepsilon$. This establishes that the total divergence term of~\nn{dsos} yields a Laurent series in~$\varepsilon$ but no $\log\varepsilon$ term.

It now remains only to study the contribution to the variation given by
$$
-\frac {d(d-2)}2\, \int_{D_\varepsilon}\, \frac{\bm{\mathcal B}\, \updelta \bm{\bar\sigma}}{ \bm{\bar\sigma}}\, .
$$ 
As discussed in Section~\ref{distributions}, we can employ $\bar\sigma$ 
 as a coordinate in a collar neighborhood of~$\Sigma$. Ignoring a finite contribution, it will be sufficient to restrict the above integral to this collar. Since~$|\nabla\bar\sigma|=1$ along $\Sigma$ for any scale $g$,  the volume form can be written as
 $$
\ext \! V^g = d\bar \sigma\,   \ext \! A(\bar\sigma)
$$
where $\ext \! A(\bar\sigma)$ is a measure for constant $\bar \sigma$ hypersurfaces $\Sigma_{\bar \sigma}$. Then by Fubini's theorem  
the collar restriction of
integral displayed above 
is (in some scale $g$ where $\bm{\mathcal B}=[g\, ;\, {\mathcal B}]$)
$$
-\frac{d(d-2)}2\, \int_\varepsilon^\star \frac{d\bar\sigma}{\bar\sigma} \int_{\Sigma_{\bar \sigma}} \ext \! A(\bar\sigma) \, \updelta \bar\sigma\, {\mathcal B}\, ,
$$
where $\star$ indicates our choice of collar neighborhood.
Noting that 
$\ext \! A|_\Sigma=\ext \! A(0)=\ext \! A^{\bar g}$, and using that  the obstruction density is non-singular along $\Sigma$  
it follows 
that the  behavior of this integral  is 
$$
-\frac{d(d-2)}2\,
\log(1/\varepsilon)\, 
 \int_{\Sigma} \ext \! A^{\bar g} \, \updelta \bar\sigma\, B+{\mathcal O}(\varepsilon^0) \, .
$$
Remembering that $\updelta \bar\sigma \stackrel\Sigma=\hat\updelta \sigma_0$,
we can read off the variation of the 
anomaly from the above display. Thus we find
that the functional gradient of the regulated volume is a Laurent series plus the desired log term:\begin{equation*}
\bm {\rm Laurent}(\varepsilon) +\frac {d(d-2)}2 \,\log\varepsilon\  \bm B\, .
\end{equation*}
Equation~\nn{varn} for the functional gradient of the anomaly follows accordingly.

\subsection{Examples}

Let $(M,g)$ be a  Riemannian $d$-manifold. Since we are given a metric~$g$ as data, we may define a true scale $\bm \tau=[\hh g\, ;\, 1]$. Now suppose we are further given a hypersurface $\Sigma$ as the zero locus of some function $\sigma_0:M\to {\mathbb R}$. As explained in Section~\ref{AHs}, we may improve this to a {\it unit defining function} meaning that $|\nabla \sigma|^g=1$ also away from~$\Sigma$. 
This yields a corresponding defining density for ~$\Sigma$
$$
\bm \sigma=[\hh g\, ;\,  \sigma]\, ,
$$
which, for our renormalized volume computation, we wish to further  improve  to a conformal unit defining density~$\bm {\bar \sigma}$. A closed form algorithm for this was given in~\cite{GW15}. In dimension~$d=3$  (see~\cite{GGHW15} for explicit expressions in dimensions $d=4,5$)
the algorithm gives $\bm{\bar \sigma}=[\hh g\, ;\, \bar \sigma]$ where (here $n:=\nabla\sigma$ rather than $\nabla\bar\sigma$)
$$
\bar \sigma =  \sigma\, \Big(1+\frac\sigma4\, 
\nabla.n+\frac{\sigma^2}{12}\, \Big[
2\, (\nabla.n)^2+\nabla_n\nabla.n
+4\, \J\, 
\Big]
\Big)\, .
$$
An elementary computation shows that the $\bm {\mathcal S}$-curvature of the above conformal unit defining scale $\bm{\bar \sigma}$ is
$$
\bm {\mathcal S}=1-\frac{
\ \bm{\sigma}^3}{12}\, \Big[\hh g\, ;\,
2\, \Delta\nabla.n+2\, \nabla_n^2\, \nabla.n+8\, (\nabla.n)\, \nabla_n\nabla.n+3\, (\nabla.n)^3+8\, \nabla.n\, \J + 8\, \nabla_n\J\,  \Big]\, .
$$
Then a simple calculation based on the above formula~\cite{CRMouncementCRM}---or a general holographic formula, or a general recursion (see~\cite{GW15})---gives the obstruction density for surfaces in terms of the extrinsic BGG operator of Equation~\nn{dualBGG}
\begin{equation}\label{surfaceB}
-3 \bm B = \D_{ab}^*\,  \bm{\IIo}^{ab}\, .
\end{equation}
A formula for the generally conformally  curved surface obstruction density was first found in~\cite{ACF} (see also~\cite{YuriThesis} for a related result); this
reduces to the Euclidean result~\nn{EWillmore} when the host metric is conformally flat. The two-dimensional obstruction density~$\bm B$ is well known
to be the functional gradient of the Willmore energy 
 $-\frac16\int_\Sigma\bm K=-\frac 16\int_\Sigma  \bm{\IIo}_{ab}\, \bm{\IIo}^{ab}$. 
Since the Euler characteristic does not contribute to the functional gradient, this establishes that the variation of the anomaly~${\mathcal A}$ in Equation~\nn{Willd=2}
is given by $\frac{3}{2}\bm B$ in accordance with Equation~\nn{varn}.

In dimension $d=4$,   the obstruction density~$\bm B$ was computed explicitly in~\cite{GGHW15} by using the holographic formula of~\cite[Theorem 8.11]{GW15}:
\begin{equation*}
\begin{split}
\bm B\, =\, \frac16\Big[
\D_{ab}^* \big(2\, \bm{\IIo}^{c(a}\, \bm{\IIo}^{b)\circ}_c
&+\bm{\mathcal F}^{(ab)\circ}\big)
-\ \bm{\IIo}^{ab}\! {\bm B}_{ab}
\\[1mm]&
+\frac12 {\bm K}^2
+\bm{\mathcal F}^{(ab)\circ}
\big(\, 
\bm{\IIo}^c_a\bm\IIo_{bc}^{\phantom{c}}+2
\bm{\mathcal F}_{ab}
\big)
+
(\bm{\hat n}^d\bm W_{\!dabc})^{\!\top} \bm{\hat n}^e \bm W_{\!e}{}^{abc}
\Big]\, .
\end{split}
\end{equation*}
This  density was proved to be  the functional gradient of  $\frac 16\int_\Sigma \bm{\IIo}_{ab}\, \bm{\mathcal F}^{ab}$ (see~\cite[Proposition~1.2]{GGHW15}).
For compactly supported variations, the double divergence term in the three dimensional extrinsic $Q$-curvature formula~\nn{d=3Will} does not contribute to the functional gradient. Hence the variation of the $d=4$ anomaly ${\mathcal A}$ is  precisely $4\bm B$, which once again agrees with Equation~\nn{varn}.

\section*{Acknowledgements}

We would like to thank Robin Graham for showing us the details of his proof that the obstruction density is variational. This work would not have been possible without his input.
 Both authors gratefully acknowledge support from the Royal Society of New Zealand via Marsden Grant 13-UOA-018 and A.W. was  supported in part by a Simons Foundation Collaboration Grant for Mathematicians ID 317562.

\appendix

\section{Explicit metrics}

Given an explicit metric and hypersurface
$$
ds^2=g\, ,\quad \Sigma={\mathcal Z}(\sigma)\, ,
$$
and a choice of scale $\bm \tau=[\hh g\, ; \, \tau]$, with the aid of computer software (see for example~\cite{grtensor})  it is not difficult to 
calculate the divergences and anomaly for the regulated volume
for a singular metric determined by a asymptotic solution to the singular Yamabe 
as described in Theorem~\ref{BIGTHE}.
These are given by
 our formula:
\begin{equation}\label{Vreg}
\Vol_\varepsilon\ =\ \sum_{k=1}^{d-1}\, 
 \frac{(k-1)!\, \int_\Sigma (-\D)^{d-k-1}\, \frac1{\scalebox{.9}{${\bm \tau}$}^k}}
 {(d-k-1)!\, (d-2)!\, k} \ \frac{1}{\varepsilon^k} \ + \ 
  \frac{\int_\Sigma (-\D)^{d-1}\log{\bm \tau}}{(d-1)!\, (d-2)!}\ \log\varepsilon
 \ + \ {\mathcal O}(\varepsilon^0)\, .
\end{equation}
The divergences will in general depend on the choice of true scale ${\bm \tau}$ while the anomaly given by the coefficient of the logarithm is a 
conformal invariant.
Given~$g$, a natural choice for the scale is $\bm \tau=[\hh g\, ;\, 1]$.
We will compute
the above formula in that scale for explicit four and five dimensional hypersurfaces~$\Sigma$.

\subsection{The Kasner metric}\label{Kasner}

The Kasner metric models spatially inhomogeneous expanding cosmologies; see for example~\cite[Chapter 14]{Landau}. Consider the following metric and hypersurface:
\begin{equation}\label{KM}ds^2=dt^2 + t^{2\alpha} dx^2
+t^{2\beta} dy^2 + t^{2\gamma}dz^2\, ,\quad \Sigma\subset{\mathcal Z}(t-1)\, .\end{equation}
Here $\Sigma$ is some bounded region in the $t=1$ coordinate slice. 
Thus, in this example the hypersurface~$\Sigma$ is not closed, and {\it a priori} the anomaly and divergences can acquire contributions integrated along $\partial \Sigma$, arising from the divergence term in the integration by parts result of Theorem~\ref{BIGTHE}. In fact, by choosing a bulk integration region intersecting~$\Sigma$ orthogonally, these terms vanish for this example. Again we defer a detailed study of these terms
to a sequel article.
We work in Euclidean signature but it is not difficult to extend our results to the physical Lorentzian signature in which $t$ becomes a time coordinate and $\Sigma$ is a spatial region.  

The mean curvature and the traced-square of the second fundamental form for the hypersurface~$\Sigma$ have simple expressions in terms of the parameters $(\alpha,\beta,\gamma)$:
$$
H=\frac{\alpha+\beta+\gamma}{3}\, \qquad
\II^2:=\II_{ab}\II^{ab}=\alpha^2+\beta^2+\gamma^2\, .
$$
Also, the rigidity density of~$\Sigma$ is given by
$$
K:=\IIo_{ab}\, \IIo^{ab}=\II^2-3H^3\, .
$$
Note that along~$\Sigma$, the scalar curvature obeys
$$
\J\, |_\Sigma=-\, \frac{\, K+6H(2H-1)\, }{6}\, .
$$
Imposing the Kasner conditions
$$
H=0=\II^2
$$
on the parameters $(\alpha,\beta,\gamma)$, the metric~$g:=ds^2$ becomes the Ricci-flat Kasner metric, 
but for added generality, we relax these conditions in the following computation.

Denoting $\sigma=t-1$, we introduce the defining density $\bm \sigma=[\hh g\, ,\, \sigma]$ which can, according to Theorem~\ref{BIGTHE}, be improved to a conformal unit defining density~$\bar \sigma$. An explicit recursion for finding $\bm \sigma=[\hh g\, ;\, \bar \sigma]$ is given in~\cite[Proposition 4.9]{GW} (see also the examples in~\cite[Appendix A]{GGHW15}). Applying this recursion we find  
\begin{equation}\label{bs}\bar\sigma=
\sigma\Big(1\, +\, \frac{H}{2}\ \sigma\, 
+\, \Big[
\frac{J|_\Sigma}{6}+\frac{H(H-1)}{2}
\Big]\ \sigma^2
\, -\,
\Big[
\frac{(5H-6)K}{72}-\frac{H(H-2)(H-3)}{24}
\Big]\sigma^3
\Big)\, .
\end{equation}
The corresponding $\bm{\mathcal S}$-curvature obeys 
$$\bm{\mathcal S}=1+\bm \sigma^4\, [\hh g\, ;\, B+{\mathcal O}(\sigma)]\, ,$$
with obstruction density given by $\bm B=[\hh g\, ; B]$ where
$$
B=\frac{K(H-1)^2}{4}\, .
$$
Choosing the true scale ${\bm \tau}=[\hh g\,  ;\,  1]$ and using Equations~\nn{Ldef} and~\nn{Dlog} it is not difficult to compute
$$
\D \frac1{\tau^2}\ \Big|_\Sigma=-4H\, , \qquad
\D^2 \frac1\tau\ \Big|_\Sigma=-2\J\, |_\Sigma
\, ,\qquad
\D^3 \log\tau\, \big|_\Sigma=4K(H-1)\, .
$$
Hence, using Equation~\nn{Vreg}, we have
\begin{equation}\label{Kasvol}
\Vol_\varepsilon=A_\Sigma^{\bar g}\, \Big(
\frac{1}{3\varepsilon^3}\ +\ \frac{H}{\varepsilon^2}\ -\ \frac{\J\, |_\Sigma}{2\varepsilon}\ - \
\log\varepsilon \  \frac{K(H-1)}{3}
\Big)+{\mathcal O}(1)\, ,
\end{equation}
where $A_\Sigma^{\bar g}$ is the area of the hypersurface~$\Sigma$.
This equation should be compared with our general result for spaces embedded in four-manifolds in~\nn{d=3Will}.

As a final check on this result, given
 the simplicity of the Kasner-type metric in Equation~\nn{KM}, we can compute the integral defining the regulated volume in Equation~\nn{vregd} by brute force. In particular, we wish to compute 
\begin{equation}\label{vole}
\Vol_\varepsilon=\int_{D_\varepsilon} \frac{\sqrt{\det g}}{\bar \sigma^4}\, .
\end{equation}
For simplicity, we take $D_\varepsilon$ to be  the volume determined by the solid coordinate cylinder 
$$
 \{(t,x,y,z)\, |\, (x,y,z)\in \Sigma\,,\: \varepsilon \leq \bar\sigma(t)\mbox{ and } t< R\}\, .
$$
This corresponds to the volume of the dark gray trumpet-shaped space-time region depicted below:
\vspace{-.1cm}
\begin{center}
\includegraphics[scale=.32]{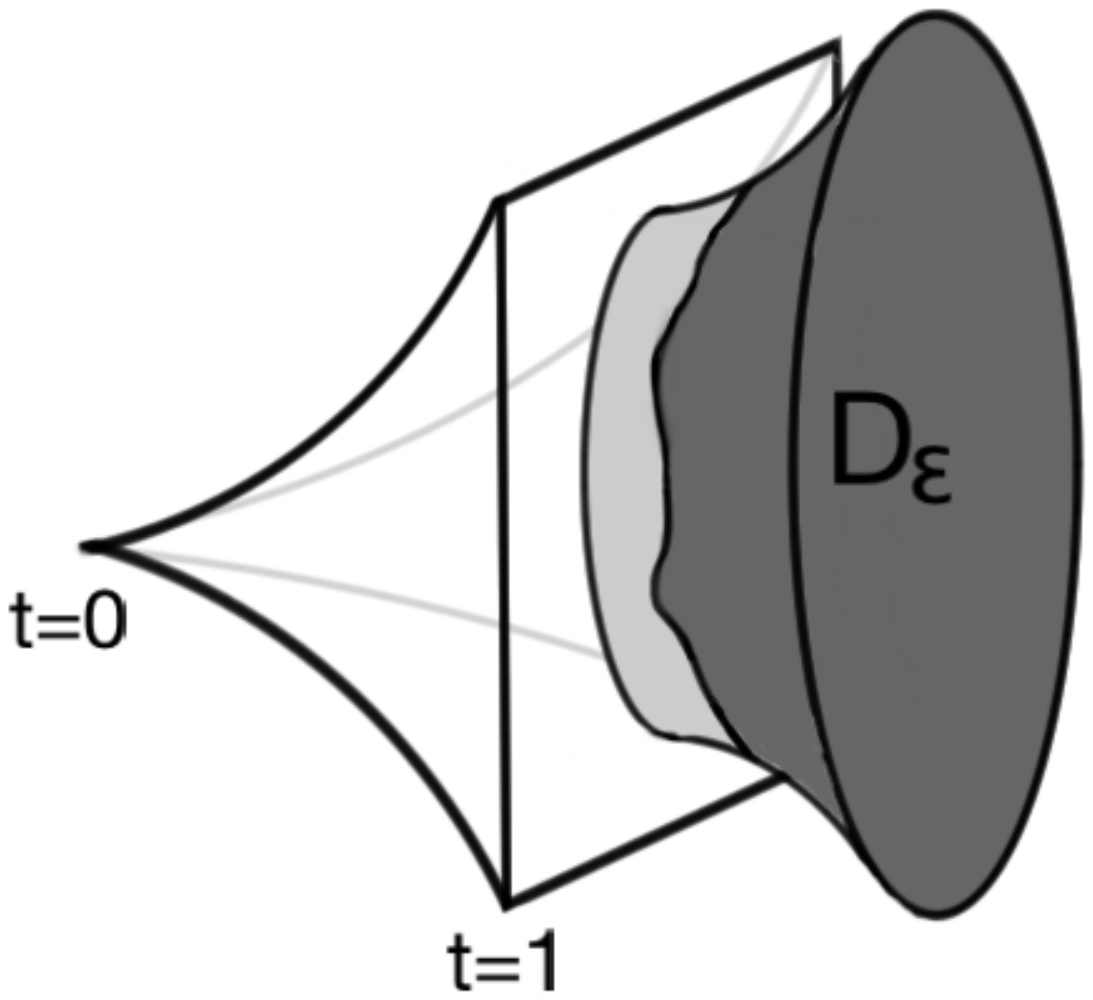}
\end{center}
\vspace{-5.3cm}
To compute the Laurent series expansion in~$\varepsilon$ of the integral in Equation~\nn{vole}, we expand the integrand in powers of $\sigma=t-1$ and find
$$
\frac{\sqrt{\det g}}{\bar\sigma^4}=
\frac1{(t-1)^4}\ +\ \frac{H}{(t-1)^3}\ +\ 
\frac{K+3H(H-3/2)}{9(t-1)^2}\ + \ \frac{K(H-1)}{3(t-1)}\ +\ {\mathcal O}(1)\, .
$$
We must also solve
 $$\bar \sigma(t_0) = \varepsilon\, ,$$
 with $\bar \sigma$ given by Equation~\nn{bs}, for the starting point of the $t$-integral
 as a power series in~$\varepsilon$. For that we find
  $$
 t_0=1+\varepsilon+
 \frac{H}{2}\, \varepsilon^2\, +\, 
 \frac{K+12H(H+1)}{36}\, \varepsilon^3
 \, + \, {\mathcal O}(\varepsilon^4)\, .
 $$
 Assembling the above data, the integration over $t$ in Equation~\nn{vole} is easy to perform and gives
$$
 \int_{t_0}^R dt\ \frac{\sqrt{\det g}}{\bar\sigma^4}=\frac1{3\varepsilon^3}\, +\, \frac{H}{\varepsilon^2}\, -\, 
 \frac{\J\, |_\Sigma}{2\varepsilon}
 -\, \frac{K(H-1)}{3} \, \log\varepsilon\,+{\mathcal O}(1)\, .
 $$ 
 This matches perfectly the regulated volume expression~\nn{Kasvol}.

\subsection{Generalized Hawking energies}

The Hawking energy associated to a compact spatial region with boundary~$\Sigma$ depends on the integral of mean curvature squared $\int_\Sigma H^2$. For conformally flat structures, this quantity recovers the Willmore energy of $\Sigma$. Therefore it is interesting to wonder whether the higher dimensional generalizations of the Willmore functional provided by the anomaly~${\mathcal A}$ are relevant to the problem of constructing quasi-local conserved quantities for general relativity in dimensions greater than four. We will not consider this problem any further except as motivation to compute the regulated volume for spatial regions of a six dimensional Schwarzschild black hole.

The six dimensional Schwarzschild metric is given by
$$
-\, \Big(\, 1-\, \frac{\, r_{\rm s_{\phantom{.\!}}}^3}{\, r^3}\, \Big)\, dt^2 + ds^2\, ,
$$
where the Euclidean signature spatial metric 
$$
ds^2=\frac{dr^2}{1-\frac{r_{\rm s_{\phantom{.\!}}}^3}{r^3}}+r^2 d\Omega^2\, ,
$$
and  $d\Omega^2$ is the metric for a round 4-sphere.
We take as data for our regulated volume the pair
$$
g=ds^2\, ,\qquad \Sigma={\mathcal Z}(r-r_0)\, .
$$
Here~$\Sigma$ is the closed hypersurface given by a 4-sphere of radius $r_0>r_{\rm s}$. We then consider the regulated volume of a bounded region $D$ with inner boundary $\Sigma$.

The hypersurface $\Sigma$ is
umbilic (vanishing trace-free second fundamental form) with mean curvature
$$
H=\frac{\sqrt{1-\frac{r_{\rm s_{\phantom{.\!}}}^3}{r_0^3}}\, }{r_0}\, .
$$
The metric $ds^2$  has vanishing (and therefore constant) scalar curvature $\J=0$.
However, the hypersurface~$\Sigma$ is not a conformal infinity of $ds^2$ so this metric does not solve the our singular Yamabe problem. Indeed
$$
\bar\sigma=H r_0^2 s \Big(
1
- \frac{\scriptstyle 5\mu-2}{\scriptstyle4}\, s
+ \frac{\scriptstyle\mu(5\mu+22)}{\scriptstyle24}\, s^2
+\frac{\scriptstyle\mu(5\mu^2-154\mu-256)}{\scriptstyle192}\,s^3
+\frac{ \scriptstyle\mu (3\mu^3+50\mu^2+944\mu+704)}{\scriptstyle384}\, s^4\Big)\, ,$$
where 
$$
s:=\frac{r-r_0}{H^2 r_0^3}\,  \ \ \mbox{ and }\ \ \mu:=\frac{r_{\rm s}^3}{r_0^3}\, ,
$$
determines a conformal unit defining density $\bm \sigma=[\hh g\, ;\, \bar \sigma]$. Moreover, we find that the corresponding $\bm {\mathcal S}$-curvature obeys
$$
\bm{\mathcal S}=1+  {\mathcal O}(\bm \sigma^6)\, ,
$$
so that the obstruction density vanishes.
This implies that the surface $\Sigma$ is a critical point of the generalized Willmore functional~${\mathcal A}$.

Once again, choosing  the true scale ${\bm \tau}=[\hh g \, ; \, 1]$ and using Equations~\nn{Ldef} and~\nn{Dlog}, we can compute the local terms appearing in divergences and the anomaly:
$$
\D \frac1{\tau^3}\ \Big|_\Sigma=-\, 9H\, , \quad\!\!\!
\D^2 \frac1{\tau^2}\ \Big|_\Sigma\!=6\Big(H^2+\frac{2}{r_0^2}\Big)
\, ,\quad\!\!\!
\D^3 \frac1\tau\ \Big|_\Sigma\!=6H\Big(H^2-\frac{4}{r_0^2}\Big)\, ,
\quad\!\!\!
\D^4 \log\tau\, \big|_\Sigma=-\, \frac{54}{r_0^4}\, .
$$
Equation~\nn{Vreg}
then gives the regulated volume
$$
\Vol_{\varepsilon}=
\frac{8\pi^3}{3}\Big(\frac{r_0^4}{4\varepsilon^4}
\ +\ \frac{Hr_0^4}{\varepsilon^3}
\ +\ 
\frac{r_0^2(H^2r_0^2+2)}
{4\varepsilon^2}
\ + \
\frac{Hr_0^2(H^2r_0^2-4)}{6\varepsilon}
\Big)-\pi^3\log\varepsilon+\mathcal O(1)\, .
$$
The coefficients of the four divergences above match our general results given in 
Equations~\nn{Vreg},~\nn{nnlo} and Appendix~\ref{nnnlo}.

\section{Nnnlo divergence}\label{nnnlo}

The nnnlo divergence for the case of a 
 conformal unit defining density
 is given, according to Equation~\nn{divergences}, in dimension $d\geq 5$  by
$$
-\, \frac{\int_\Sigma \D^3 \bm \tau^{4-d}}{3!(d-2)(d-3)(d-4)^2\, \varepsilon^{d-4}}\, .
$$
The main ingredients required to compute  $\D^3 \bm \tau^{4-d}\, |_\Sigma$ were given in~\cite{GW}.

We work in the scale $\bm \tau=[\hh g\, ;\; 1]$ and first  use Equation~\nn{Ldef} to compute one power of the Laplace--Robin operator
$$
\D \bm \tau^{4-d}=\big[\hh g\, ;\, 
(d-4)\big((d-6)\rho+\bar \sigma \J\, \big)
\big]\, .
$$
Thus 
\begin{eqnarray*}
\D^2 \bm \tau^{4-d}
\ =\  -(d-4)\Big[\hh g\!\!\!\! &\!\! ; \!\!& \!\!
(d-4)\Big((d-6)\big(\nabla_n\rho-(d-3)\rho^2\big)+\J\, \Big)\\[1mm]
&\!\!+\!\!&
\!\!
\bar\sigma\big((d-2)(\nabla_n-(d-2)\rho)\J\, 
+(d-6)\Delta\rho\big)
\ +\ {\mathcal O}(\bar\sigma^2)\, \Big]\, .
\end{eqnarray*}
In the  above we used that for a conformal unit defining density $n^2=1-2\rho\bar\sigma+\mathcal O(\bar\sigma^d)$ and that $\nabla.n=-d\rho-\bar\sigma \J$.
In turn
\begin{eqnarray*}
\D^3 \bm \tau^{4-d}
\ \stackrel\Sigma=\ (d-2)(d-4)
\Big[\hh g\! \!\!&\!\! ;\!\! &\!\! 
(d-4)(d-6)\big(\nabla_n^2\rho
+(3d-8)H\nabla_n \rho
-(d-2)(d-3)H^3
\big)
 \\
&\!\!+\!\!&\!\!
2(d-3)\big(\nabla_n\J+(d-2) H\J\, \big) 
+(d-6)\Delta\rho
\Big]\, .
\end{eqnarray*}
Here we have again used the aforementioned conformal unit defining density properties as well as Lemma~\ref{Hlemma}. By virtue of the second identity in Equation~\nn{box} we have
\begin{eqnarray*}
\D^3 \bm \tau^{4-d}\!
\ \stackrel\Sigma=\ (d-2)(d-4)
\Big[\hh g\! \!\!&\!\! ;\!\! &\!\! \!\!
(d-3)(d-6)\big(\nabla_n^2\rho
+(3d-10)H\nabla_n \rho
-(d-2)(d-4)H^3
\big)
 \\
&\!\!+\!\!&\!\!
2(d-3)\big(\nabla_n\J+(d-2) H\J\, \big) 
-(d-6)\bar\Delta H
\Big]\, .
\end{eqnarray*}
Now we employ Equations~\nn{1},~\nn{2} and~\nn{3} to obtain the required result:
\begin{equation*}
\begin{split}
\D^3  \bm \tau^{4-d}\!
  \stackrel\Sigma= (d\!-\!4)\Big[\hh g\, ;\!
&- 2\, (d\!-\!6)\, \Big(
\nablab^a\nablab^b \IIo_{ab}-(d\!-\!3)(d\!-\!4)H\big((d\!-\!2)\Rho_{ab}\hat n^a\hat n^b+K\big)\Big)
\\&
-(d\!-\!2)(d\!-\!6)\, \Big((d\!-\!2)\ \IIo^{ab}{\mathcal F}_{ab}+(d\!-\!4)\big(\IIo^{ab}\bar\Rho_{ab}
+\nablab^a(\Rho_{ab}\hat n^b)^{\!\top}\big)
\Big)
\\&
+(d\!-\!2)(d\!-\!3)\big(\hat n^a \nabla_a \J\,  +\, (3d\!-\!10)H\J\, 
 -\, (d\!-\!2)(d\!-\!4)(d\!-\!6)H^3
\big)
\Big]\, .
\end{split}
\end{equation*}

\newcommand{\msn}[2]{\href{http://www.ams.org/mathscinet-getitem?mr=#1}{#2}}
\newcommand{\hepth}[1]{\href{http://arxiv.org/abs/hep-th/#1}{arXiv:hep-th/#1}}
\newcommand{\maths}[1]{\href{http://arxiv.org/abs/math/#1}{arXiv:math/#1}}
\newcommand{\mathph}[1]{\href{http://lanl.arxiv.org/abs/math-ph/#1}{arXiv:math-ph/#1}}
\newcommand{\arxiv}[1]{\href{http://lanl.arxiv.org/abs/#1}{arXiv:#1}}

\end{document}